\documentclass[11pt,letterpaper]{article}

\usepackage{amssymb,amsmath,amsfonts}
\usepackage{graphicx,xcolor,enumitem}
\usepackage{epsfig}
\usepackage{amsthm}
\usepackage{bm}
\usepackage{multirow}
\usepackage{subcaption}
\usepackage[round]{natbib}
\usepackage{csquotes}
\renewcommand{\mkbegdispquote}[2]{\itshape}
\usepackage[twoside, hmarginratio=1:1, vmarginratio=1:1, left=1in,top=1in]{geometry}

\RequirePackage[breaklinks=true, hidelinks]{hyperref}
\usepackage{breakcites}

\usepackage{algorithm, algorithmic}

\usepackage{accents}

\newcommand{\ang}[1]{\left\langle  #1 \right\rangle }
\newcommand{\cA}{\mathcal{A}}

\newcommand{\E}{\mathbb{E}}
\newcommand{\R}{\mathbb{R}}

\newcommand{\cS}{{\mathcal S}}

\newcommand{\cP}{{\mathcal P}}

\newcommand{\tr}{{\rm{tr}}}

\newcommand{\one}{\mathbf{1}}

\newcommand{\cX}{{\mathcal X}}
\newcommand{\cY}{{\mathcal Y}}

\newcommand{\cM}{{\mathcal M}}

\newtheorem{theorem}{Theorem}

\newtheorem{definition}[theorem]{Definition}

\newtheorem{lemma}[theorem]{Lemma}

\newtheorem{proposition}[theorem]{Proposition}
\newtheorem{remark}[theorem]{Remark}

\theoremstyle{definition}

\numberwithin{equation}{section}
\numberwithin{theorem}{section}

\begin{document}

\title{The McCormick martingale optimal transport}

\author{
	Erhan Bayraktar \thanks{Department of Mathematics, University of Michigan, Ann Arbor, Email: erhan@umich.edu. Erhan Bayraktar is partially supported by the National Science Foundation under grant DMS-2106556 and DMS-2507940, and by the Susan M. Smith chair.}
	\and Bingyan Han \thanks{Thrust of Financial Technology, The Hong Kong University of Science and Technology (Guangzhou), Email: bingyanhan@hkust-gz.edu.cn. Bingyan Han is partially supported by the National Natural Science Foundation of China (grant 12401621) and the Guangzhou-HKUST(GZ) Joint Funding Program (2024A03J0630).}
	\and Dominykas Norgilas \thanks{Department of Mathematics, North Carolina State University, Email: dnorgil@ncsu.edu.}
}

\date{February 25, 2026}
\maketitle

\begin{abstract}
    Martingale optimal transport (MOT) often yields broad price bounds for options, constraining their practical applicability. In this study, we extend MOT by incorporating causality constraints among assets, inspired by the nonanticipativity condition of stochastic processes. This, however, introduces a computationally challenging bilinear program. To tackle this issue, we propose McCormick relaxations to ease the bicausal formulation and refer to it as McCormick MOT. The primal attainment and strong duality of McCormick MOT are established under standard assumptions. Empirically, we apply McCormick MOT to basket and digital options. With natural bounds on probability masses, the average price reduction for basket options is approximately 1.08\% to 3.90\%. When tighter probability bounds are available, the reduction increases to 12.26\%, compared to the classic MOT, which also incorporates tighter bounds. For most dates considered, there are basket options with suitable payoffs, where the price reduction exceeds 10.00\%. For digital options, McCormick MOT results in an average price reduction of over 20.00\%, with the best case exceeding 99.00\%.
	\\[2ex] 
	\\[2ex]
\end{abstract}

\section{Introduction}
In the realm of option pricing, a crucial task is to establish model-independent arbitrage-free price bounds for derivatives. These bounds serve a direct application in identifying potential arbitrage opportunities in traded options. \cite{hobson2012robust} pioneered the study on robust bounds for forward start options, while Martingale Optimal Transport (MOT), introduced by \cite{beiglbock2013model} and \cite{GHLT2014MOTcontinuous} (in discrete and continuous time settings, respectively), presented a comprehensive framework utilizing tools from Monge--Kantorovich mass transport. The strong duality result has a financial interpretation, yielding sub/super hedging portfolios.  Furthering the MOT research line, \cite{beiglbock2017complete,beiglbock2019dual} considered dual attainment problems. \cite{beiglbock2016problem} identified a martingale coupling reminiscent of the classic monotone quantile coupling. Several numerical schemes for MOT were proposed by \cite{guo2019mot,eckstein2021robust,eckstein2021computation}, employing techniques such as linear programming (LP) or neural networks. Investigations into MOT stability were carried out by \cite{backhoff2022stability,bruckerhoff2022,wiesel2023continuity}.

A significant hurdle in MOT arises from the wide price bounds, thereby constraining their practical applicability. This claim is empirically substantiated by \citet[Figure 1]{beiglbock2013model} and \citet[Tables 4.2 and 4.3]{eckstein2021robust}. Achieving more precise bounds typically requires the incorporation of additional information into the modeling framework. For example, \cite{lutkebohmert2019tightening} argued that tighter bounds can be derived when more information about the variance of returns is available, although at the cost of sacrificing model independence. \cite{fahim2016model} enhanced precision of price bounds by integrating liquid exotic options into the hedging portfolio. Furthermore, \citet[Section 4]{eckstein2021robust} refined the bounds by assuming knowledge of the marginals at earlier maturities. For further techniques on tightening robust price bounds or developing computational methods for MOT, we refer readers to \cite{ansari2024improved,eckstein2021martingale,neufeld2021model,neufeld2023model} and the references therein.

In this study, we introduce another methodology inspired by recent strides in causal optimal transport (COT) \cite{lassalle2013causal,backhoff2017causal}. COT posits a crucial condition: given the past of one process $X$, the past of another process $Y$ should be independent of the future of $X$ under the transport plan. Essentially, COT imposes the nonanticipativity condition. The versatility of COT is evidenced by its applications in various domains such as mathematical finance \cite{backhoff2020adapted}, mean field games \cite{acciaio2021cournot,backhoff2023dynamic}, and stochastic optimization \cite{pflug2012distance,pflug2014multistage,acciaio2020causal}.

In Section \ref{sec:causal}, we present the causality constraint in a two-asset setting. An essential challenge arises as the MOT incorporating causality transforms into a bilinear program that is non-convex. This inherent difficulty stems from the fact that only marginals at each maturity are known in MOT. However, causality requires conditional kernels that rely on the joint distribution across different maturities, detailed in Proposition \ref{prop:equi}. First, Section \ref{sec:bicausalMOT} presents several standard results for the bicausal MOT in the finite discrete case, including primal attainment, SDP relaxations, and dual problems. Next, to address the computational challenges, we relax the original problem and consider McCormick envelopes \citep{mccormick1976} for the bilinear equality constraints. We refer to it as the McCormick MOT. Compared to SDP relaxations and other optimization techniques in \cite{audet2000branch,sherali2002enhancing,anstreicher2009semidefinite,aspremont2003relaxations,vandenberghe1996semidefinite}, McCormick relaxations offer a straightforward implementation and lead to an LP problem, making the theoretical properties also more tractable.

In discrete scenarios, McCormick MOT is a finite-dimensional LP problem and thus straightforward. Our main technical contributions are for the continuous case with absolutely continuous densities. To derive McCormick envelopes, we assume that lower and upper bounds for densities are available, which are referred to as capacity constraints in \cite{korman2015optimal}. The primal attainment is established in Theorem \ref{thm:primal} after selecting a suitable weak topology. Analogously to the classic MOT, a natural question is the strong duality of the McCormick MOT. Lemma \ref{lem:capa} and Theorem \ref{thm:dual} prove the relevant result under standard assumptions in the MOT literature. The main idea underlying these proofs remains rooted in the Fenchel-Rockafellar duality theorem.

In our numerical study, we introduce a calibration method for risk-neutral densities that ensures convex order while addressing bid-ask spreads and multiple tenors. Given the discrete nature of empirical data, our approach formulates a finite-dimensional LP problem, offering a straightforward optimization process. The improvement in price bounds depends on specific option payoffs and, in some cases, the liquidity of European options on the underlying assets. We primarily analyze basket and digital options. With the natural lower and upper bounds implied by marginal constraints, McCormick relaxations reduce the price gap for equally weighted basket options by an average of 1.08\% for stocks with liquid option markets and 3.90\% for those with moderately liquid markets. When tighter probability bounds are applied, Section \ref{sec:tighter_bnds} reports an average reduction of 12.26\%, compared to the classic MOT, which also incorporates tighter bounds. Furthermore, Section \ref{sec:arbitrary_w} shows that for most dates considered, there exists a basket option with a price reduction exceeding 10.00\%. Additionally, for most basket options studied, there exist marginals that lead to a price reduction greater than 10.00\%. For digital options, Section \ref{sec:digital} finds that McCormick MOT achieves an average price reduction of over 20.00\%, with the best case exceeding 99.00\%. While McCormick MOT may not always improve price bounds for digital options, when it does, the reduction is substantial. Besides, our methodology is {\it universal} since it integrates seamlessly with other knowledge available to the agent. The source code is publicly available at \url{https://github.com/hanbingyan/McCormick}.

The rest of the paper is organized as follows. Section \ref{sec:form} introduces the notation and formulation of MOT. Section \ref{sec:causal} discusses the motivation of the causality constraint  and its implications. Section \ref{sec:bicausalMOT} presents several theoretical results for bicausal MOT in the finite discrete case. In Section \ref{sec:McCormick}, we demonstrate the McCormick relaxations in both discrete and continuous cases. Section \ref{sec:num} shows the improvements of the price bounds with illustrative examples and empirical data. The appendix gives proofs of the main results.

\section{Problem formulation}\label{sec:form}
We consider a financial market comprising two risky assets, denoted $X$ and $Y$, while the extension to include more assets is straightforward. Suppose that the market has European options on each of the two stocks for the maturities $\{T_1, \ldots, T_N\}$, where $N \geq 2$ and the maturities are not necessarily evenly spaced. Let $X_t := X_{T_t}$ and $Y_t := Y_{T_t}$ represent the prices of risky assets at maturity $T_t$. The range of the first risky asset $X$ at time $T_t$ is denoted as $\cX_t$, assumed to be a closed subset of $\R$. Consequently, for each $t\in\{1,...,N\}$, $\cX_{1:t} : = \cX_1 \times \ldots \times \cX_t$ forms a closed subset of $\R^{t}$. We set $\cX:=\cX_{1:N}$. The set of all Borel probability measures on $\cX_t$ is denoted as $\cP(\cX_t)$. Similarly, for the second risky asset $Y$, we introduce another closed sets $\cY_{1:t} = \cY_1 \times \ldots \times \cY_t$, with the corresponding notations defined analogously.

In this study, we assume that individual risk-neutral probabilities denoted $X_t \sim \mu_t \in \cP(\cX_t)$ and $Y_t \sim \nu_t \in \cP(\cY_t)$, are known. Specifically, if the prices of European options are available for all strikes, \cite{breeden1978prices} provided a formula for the risk-neutral probability density. In cases involving a finite number of strikes, a linear interpolation method is described in \cite{davis2007range}. Additionally, an alternative cubic-spline method has been proposed by \cite{fengler2009arbitrage} for finitely many options prices. \citet[Corollary 3.15]{neufeld2022numerical} introduces an approximation method with theoretical convergence guarantees.

However, the joint distribution among different maturities or assets is typically unknown. Let $\Pi(\bar{\mu}) = \Pi(\mu_1, \ldots, \mu_N)$ represent the set of couplings $\mu \in \cP(\cX_{1:N})$ with $\mu_t$ as marginals for each $t$; $\Pi(\bar{\nu}) = \Pi(\nu_1, \ldots, \nu_N)$ is defined similarly. Consider $\Pi(\bar{\mu}, \bar{\nu})$ as the set of couplings $\pi \in \cP(\cX_{1:N} \times \cY_{1:N})$ with $\mu_t$ and $\nu_t$ as marginals. The model-independent and arbitrage-free framework in \cite{beiglbock2013model} involves examining a measure $\pi \in \Pi(\bar{\mu}, \bar{\nu})$ that satisfies the martingale condition:
\begin{equation}\label{eq:mart}
	\E_\pi [X_{t+1} | X_{1:t}, Y_{1:t}] = X_t,  \quad	\E_\pi [Y_{t+1} | X_{1:t}, Y_{1:t}]  = Y_t, \quad 1 \leq t \leq N-1.
\end{equation} 
For the sake of simplicity, we assume zero interest rates and no dividends in \eqref{eq:mart}, while our empirical results utilize forward prices to incorporate these factors. Denote $\cM(\bar{\mu}, \bar{\nu}) \subset \Pi(\bar{\mu}, \bar{\nu})$ as the set of all martingale measures satisfying marginal constraints. Suppose $\{\mu_t\}^N_{t=1}$ possess identical finite first moments and increase in convex order in $t$, that is, for all convex functions $f$ and $1 \leq t \leq N-1$,
\begin{align*}
	\int f d\mu_t \leq \int f d\mu_{t+1}. 
\end{align*}
Denote the convex order condition as $\mu_t \preceq_{cx} \mu_{t+1}$. Similarly, assume $\{\nu_t\}^N_{t=1}$ have the same finite moments and $\nu_t \preceq_{cx} \nu_{t+1}$. By \cite{strassen1965existence}, the assumptions above are necessary and sufficient conditions for $\cM(\bar{\mu}, \bar{\nu}) \neq \emptyset$.

In the subsequent discussion, we examine an exotic option with the payoff function $c(x_{1:N}, y_{1:N})$. The MOT framework \cite{beiglbock2013model,eckstein2021robust} investigates model-free bounds for this exotic option:
\begin{equation}\label{eq:MOT}
	\sup_{\pi \in \cM(\bar{\mu}, \bar{\nu})} \int c(x_{1:N}, y_{1:N}) \pi(dx_{1:N}, dy_{1:N}) \; \text{ and } \inf_{\pi \in \cM(\bar{\mu}, \bar{\nu})} \int c(x_{1:N}, y_{1:N}) \pi(dx_{1:N}, dy_{1:N}).
\end{equation} 
\eqref{eq:MOT} is known as the primal formulation. The dual problem unveils semi-static hedging strategies as outlined in \citet[Theorem 1.1]{beiglbock2013model}.

In the absence of modeling assumptions, the gap between lower and upper bounds in MOT can be substantial, thereby limiting its practical applicability. A common approach to address this challenge involves incorporating additional information. For instance, \cite{fahim2016model} augmented the model with market prices of specific exchange-listed exotic options, such as standardized digital and barrier options. Another strategy entails including information related to volatility. Joint calibration on SPX and VIX options leverages the structural link between the S\&P 500 and VIX indices; see \cite{guyon2017bounds,de2015linking} for more details. Additionally, \cite{bayraktar2021transport} examined MOT examples with volatility constraints, while \cite{lutkebohmert2019tightening} considered supplementary information regarding the variance of returns.

The conventional optimal transport framework fails to account for the temporal structure and information flow inherent in time series data. To address this limitation, recent research has introduced COT \cite{lassalle2013causal,backhoff2017causal}. When dealing with options on multiple assets, we demonstrate that the nonanticipativity (causality) condition provides an alternative means of tightening bounds, typically without requiring new data or estimation. Section \ref{sec:causal} presents the causality constraint and its implications for MOT.

\section{Motivation and implication of causality}\label{sec:causal}
\subsection{Nonanticipativity condition}
Under a risk-neutral measure $\pi$ with discrete-time processes, the underlying assets $X$ and $Y$ are typically represented as:
\begin{equation}\label{eq:binomial}
	X_{t+1} - X_{t} = X_t W_{t+1}, \qquad Y_{t+1} - Y_{t} = Y_t Z_{t+1}. 
\end{equation}  
Here, $\{ W_t \}^{N-1}_{t=1}$ and $\{ Z_t \}^{N-1}_{t=1}$ denote random returns for $X$ and $Y$, respectively, which are the only sources of randomness in the system \eqref{eq:binomial}. A standard assumption posits that $\{ W_t \}^{N-1}_{t=1}$ are independent and identically distributed, as commonly imposed in binomial tree models. $\{ Z_t \}^{N-1}_{t=1}$ also satisfy the same assumption. However, it should be noted that for each time $1 \leq t \leq N-1$, $W_t$ and $Z_t$ are allowed to be correlated. For example, a parsimonious assumption is that the correlation between $W_t$ and $Z_t$ is a constant. Under these conditions, $X_{t+1}$ is independent of $Y_t$, conditional on $X_t$. This equivalently implies that $\pi(dx_{t+1} | x_t, y_t) = \pi(dx_{t+1} | x_t)$, since knowledge of $Y_t$ does not aid in predicting $W_{t+1}$ once $X_t$ is known. It is crucial to emphasize that $Y_t$ and $X_{t+1}$ are not independent if $X_t$ is unknown. 

Aligned with the nonanticipativity condition that originates from stochastic difference (or differential) equations, \cite{lassalle2013causal} imposes that a transport plan $\pi$ should satisfy
\begin{equation}\label{eq:causal1}
	\pi (dx_{t+1} | x_{1:t}, y_{1:t}) = \pi (dx_{t+1} |x_{1:t}), \quad  t = 1, ..., N-1, \quad \text{$\pi$-a.s.},
\end{equation}
which is equivalent to 
\begin{equation}\label{eq:causal2}
	\pi (dy_t | x_{1:N}) = \pi (dy_t | x_{1:t}), \quad  t = 1, ..., N-1, \quad \text{$\pi$-a.s.},
\end{equation}
see \citet[Remark 2.2]{eckstein2024computational}. The property \eqref{eq:causal1} (or \eqref{eq:causal2}) is known as the causality condition. A transport plan that satisfies \eqref{eq:causal1} (or \eqref{eq:causal2}) is termed {\it causal}. If we interchange the positions of $X$ and $Y$ the following condition holds:
\begin{equation}\label{eq:anticausal}
	\pi (dx_t | y_{1:N}) = \pi (dx_t | y_{1:t}), \quad  t = 1, ..., N-1, \quad \text{$\pi$-a.s.},
\end{equation}
then the transport plan is {\it anticausal}. Couplings that are both causal and anticausal are referred to as bicausal. For later use, denote $\cM_{bc}(\bar{\mu}, \bar{\nu}) \subset \cM(\bar{\mu}, \bar{\nu})$ as the set of all bicausal and martingale transport plans, with $\mu_t$ and $\nu_t$ as marginals for each $t$. As a side note, since we assume only marginals $\mu_t$ and $\nu_t$ are known, the standard result regarding the decomposition into irreducible components on the real line (for example, as given in \citet[Proposition 2.3]{beiglbock2017complete}) applies to our setting.

Motivated by the nonanticipativity condition, we propose to further restrict the MOT problem \eqref{eq:MOT} by considering bicausal plans:
\begin{equation}\label{eq:bi_MOT}
	\begin{aligned}
		& \sup_{\pi \in \cM_{bc}(\bar{\mu}, \bar{\nu})} \int c(x_{1:N}, y_{1:N}) \pi(dx_{1:N}, dy_{1:N}) \\
		& \text{ and } \inf_{\pi \in \cM_{bc}(\bar{\mu}, \bar{\nu})} \int c(x_{1:N}, y_{1:N}) \pi(dx_{1:N}, dy_{1:N}).
	\end{aligned}
\end{equation} 

Another motivation for considering bicausality in MOT arises from its connection to the no-arbitrage condition and market completeness \cite{krvsek2024general}. Suppose the classic MOT minimization problem in \eqref{eq:MOT} has an optimizer $\pi^*$, where bicausality conditions are not imposed. Let $\mu^* \in \cP(\cX_{1:N})$ and $\nu^* \in \cP(\cY_{1:N})$ be the marginals of $\pi^*$ on $\cX_{1:N}$ and $\cY_{1:N}$, respectively. Then $\pi^* \in \Pi(\mu^*, \nu^*)$. \citet[Section 3.1]{krvsek2024general} established that if the market with assets $X$ and $Y$ is complete, then the coupling $\pi^*$ satisfies the no-arbitrage condition in \citet[Definition 3.1]{krvsek2024general} if and only if the coupling $\pi^*$ is bicausal. The definition of completeness in \cite{krvsek2024general} depends on the marginals $\mu^*$ and $\nu^*$. Let $\mathbb{F}^X$ denote the canonical filtration of $X_{1:N}$. A market is said to be complete, if for every $(\mathbb{F}^X, \mu^*)$-martingale $M^X$, there exists an $\mathbb{F}^X$-adapted process $\Delta^X = \{\Delta^X_t\}^{N-1}_{t=0}$, such that
\begin{equation*}
	M^X_t = \E_{\mu^*}[M^X_1] + \sum^t_{s=1} \Delta^X_{s-1} (X_s - X_{s-1}), \quad \text{$\mu^*$-a.s.}, \quad t \in \{0, \ldots, T\}.
\end{equation*}
The same condition is imposed when replacing $X$ with $Y$ and $\mu^*$ with $\nu^*$; see \citet[Equation 3.1]{krvsek2024general} for details. Although the marginals $\mu^*$ and $\nu^*$ across different maturities are not specified in advance, this result still offers an alternative perspective for interpreting bicausality conditions.

Before presenting McCormick relaxations, we discuss some properties for the bicausal MOT \eqref{eq:bi_MOT}. The first consequence of the bicausal condition is the equivalence of the martingale property under individual filtrations and the joint filtration. The proof of Proposition \ref{prop:joint_indiv} directly applies the causality condition \eqref{eq:causal1} and is therefore omitted. Besides, the tower property can be used to establish the converse direction, showing that \eqref{eq:mart} implies \eqref{eq:mart_ind}. However, the tower property cannot generally prove that \eqref{eq:mart_ind} implies \eqref{eq:mart}.
\begin{proposition}\label{prop:joint_indiv}
	Suppose $\pi \in \Pi(\bar{\mu}, \bar{\nu})$ and the following martingale condition holds with the filtration generated by each asset:  
	\begin{equation}\label{eq:mart_ind}
		\E_\pi [X_{t+1} | X_{1:t}]  = X_t, \qquad \E_\pi [Y_{t+1} | Y_{1:t}] = Y_t, \quad 1 \leq t \leq N-1.
	\end{equation} 
	If $\pi$ is bicausal, then \eqref{eq:mart_ind} implies that the martingale condition \eqref{eq:mart} under the filtration generated by $(X, Y)$ also holds. 
\end{proposition}

\subsection{Non-convexity}
When the joint distributions and, consequently, conditional kernels $\mu(dx_{t+1}|x_{1:t})$ are provided, the causality condition \eqref{eq:causal1} (or \eqref{eq:causal2}) imposes a linear constraint on $\pi$; see \citet[Proposition 2.4 (3)]{backhoff2017causal}. However, this is no longer applicable when the conditional kernels are not predetermined. The corresponding equivalent condition is outlined in Proposition \ref{prop:equi}. The proof follows a rationale similar to that of \citet[Proposition 2.4 (3)]{backhoff2017causal}, and we include it in the Appendix for the sake of completeness.
\begin{proposition}\label{prop:equi}
	$\pi \in \Pi(\bar{\mu}, \bar{\nu})$ is causal if and only if  
	\begin{equation}\label{eq:nonconvex}
		\int h_t(y_{1:t})\left [ g_t(x_{1:N}) - \int g_t(x_{1:t}, \bar{x}_{t+1:N}) \pi(d\bar{x}_{t+1:N} | x_{1:t}) \right ] \pi(dx_{1:N}, dy_{1:N}) = 0, 
	\end{equation}
	for all $1 \leq t \leq N$, $h_t \in C_b(\cY_{1:t})$, and $g_t \in C_b (\cX_{1:N})$.
\end{proposition}

To observe that the causality constraint is no longer linear in $\pi$, one can think about the discrete case. The conditional mass function $\pi(x_{t+1:N} | x_{1:t}) = \pi(x_{1:t}, x_{t+1:N})/\pi(x_{1:t})$, where the denominator $\pi(x_{1:t})>0$ is not uniquely determined by marginals $(\mu_t,\nu_t)_{1\leq t\leq N}$. In the subsequent section, we demonstrate that this constraint is bilinear in $\pi$ and non-convex. This characteristic poses a primary challenge in the bicausal MOT and makes the program difficult to solve.

Alternatively, we can also rewrite \eqref{eq:bi_MOT} as a two-stage optimization problem. By Proposition \ref{prop:joint_indiv}, if $\mu \in \cM(\bar{\mu})$ and $\nu \in \cM(\bar{\nu})$, then $\pi \in \Pi_{bc}(\mu, \nu)$ implies that the martingale condition \eqref{eq:mart} under the filtration generated by $(X, Y)$ also holds. Hence, the minimization problem in \eqref{eq:bi_MOT} becomes
\begin{equation}\label{eq:two_stage}
	\begin{aligned}
		& \inf_{\pi \in \cM_{bc}(\bar{\mu}, \bar{\nu})} \int c(x_{1:N}, y_{1:N}) \pi(dx_{1:N}, dy_{1:N}) \\
		& = \inf_{\substack{\mu \in \cM(\bar{\mu}), \\ \nu \in \cM(\bar{\nu})}} \;\; \inf_{\pi \in \Pi_{bc}(\mu, \nu)} \int c(x_{1:N}, y_{1:N}) \pi(dx_{1:N}, dy_{1:N}) := \inf_{\substack{\mu \in \cM(\bar{\mu}), \\ \nu \in \cM(\bar{\nu})}} C(\mu, \nu).
	\end{aligned}
\end{equation} 
Here, $\cM(\bar{\mu})$ denotes the set of martingale couplings with $\mu_t$ as marginals for each $t$, and $\cM(\bar{\nu})$ is defined similarly. $\Pi_{bc}(\mu, \nu)$ represents the set of bicausal couplings with $\mu$ and $\nu$ as marginals. 

We illustrate the difficulty of primal attainment and strong duality for bicausal MOT by considering \eqref{eq:two_stage}. For the inner minimization problem, primal attainment and strong duality are obtained by \cite{backhoff2017causal}. The dual attainment is given in \citet[Theorem 4.8]{krvsek2024general}. Primal attainment for the entire problem typically requires the lower semicontinuity of $C(\mu, \nu)$ and the compactness of $\cM(\bar{\mu})$ and $\cM(\bar{\nu})$ under a suitable topology. We leave this general case for future research, noting that techniques from stability results \cite{wiesel2023continuity,backhoff2022stability,eckstein2024computational} may be useful. The outer minimization problem is non-convex, and thus, strong duality fails in general. Consequently, dual attainment becomes less relevant.

In the finite discrete case, Section \ref{sec:bicausalMOT} further explores several properties of bicausal MOT and its connection to quadratically constrained quadratic programming (QCQP). 

\section{Bicausal MOT in the finite discrete case}\label{sec:bicausalMOT} 

With finite discrete marginals, $\pi(x_{1:N}, y_{1:N})$ can be interpreted as probability mass functions. To simplify the notation, we write the marginals of $\pi(x_{1:N}, y_{1:N})$ on $\cX_{1:t}$ as $\pi(x_{1:t})$, without introducing additional subscripts. This convention applies to all other marginals as well. The causality constraint \eqref{eq:causal2} (or \eqref{eq:causal1}) can be equivalently expressed as
\begin{equation}\label{eq:bilin}
	\pi(x_{1:N}, y_t) \pi(x_{1:t}) = \pi(x_{1:t}, y_t)  \pi(x_{1:N}), \quad t = 1, \ldots, N-1.
\end{equation}
We emphasize that \eqref{eq:bilin} holds pointwise, while \eqref{eq:causal1} (or \eqref{eq:causal2}) holds $\pi$-a.s. Indeed, if $\pi(x_{1:N}) > 0$, then $\pi(x_{1:t}) > 0$. Both conditional probabilities in \eqref{eq:causal2} are well-defined. Then \eqref{eq:bilin} holds. If $\pi(x_{1:N}) = 0$, then $\pi(dy_t | x_{1:N})$ is not defined. \eqref{eq:causal2} is not required to hold. However, we note that $\pi(x_{1:N}, y_t) = 0$ since $\pi(x_{1:N}) = 0$. Therefore, \eqref{eq:bilin} is valid as $\pi(x_{1:N}, y_t) \pi(x_{1:t}) = \pi(x_{1:t}, y_t) \pi(x_{1:N}) = 0$.

\eqref{eq:bilin} represents a bilinear equality constraint on $\pi$ and the anticausal constraint \eqref{eq:anticausal} is similar. Without loss of generality, we focus on the minimization problem, as the maximization problem can be handled in the same way. The bicausal MOT in the finite discrete case is given by
\begin{equation}\label{eq:bicausal_discrete}
	\inf_{\pi \in \cM_{bc}(\bar{\mu}, \bar{\nu})} \int c(x_{1:N}, y_{1:N}) \pi(x_{1:N}, y_{1:N}).
\end{equation}
Proposition \ref{prop:bicausal_primal} establishes primal attainment for \eqref{eq:bicausal_discrete} and provides a property of the optimizers. When the supports are discrete and finite, the feasible set $\cM_{bc}(\bar{\mu}, \bar{\nu})$ becomes a subset of a Euclidean space with the usual topology. Let $\partial \cM_{bc}$ denote the boundary of $\cM_{bc}(\bar{\mu}, \bar{\nu})$. Recall that every continuous linear functional on a compact convex set attains its minimum at an extreme point; see \citet[Propositions 2.1.2 and 2.4.1]{bertsekas2009convex} and also Bauer maximum principle. Since $\cM_{bc}(\bar{\mu}, \bar{\nu})$ is non-convex, Proposition \ref{prop:bicausal_primal} only guarantees that an optimizer is on $\partial \cM_{bc}$. Besides, the characterization of the properties of  $\partial \cM_{bc}$ remains an open question.

\begin{proposition}\label{prop:bicausal_primal}
	Suppose the supports $\cX$ and $\cY$ are discrete and finite. Then the bicausal MOT problem \eqref{eq:bicausal_discrete} has an optimizer. Moreover, at least one of the optimizers is on $\partial \cM_{bc}$, the boundary of $\cM_{bc}(\bar{\mu}, \bar{\nu})$ under the Euclidean topology. 
\end{proposition}

In the optimization literature, bicausal MOT is a special case of non-convex QCQP problems \cite{aspremont2003relaxations,vandenberghe1996semidefinite}. Indeed, by flattening the multidimensional array $\pi(x_{1:N}, y_{1:N})$ into a one-dimensional vector $z$ of appropriate length, the objective function in bicausal MOT \eqref{eq:bicausal_discrete} becomes $c^\top z$, where the cost vector $c$ is also flattened. The martingale and marginal constraints are linear equalities for $z$, expressed as $q^\top_i z + r_i = 0$ with suitable vectors $q_i$ and scalars $r_i$. The causal and anticausal constraints can be rewritten as $z^\top P_i z = 0$. Therefore, our problem is a special case of the non-convex QCQP problem discussed in \cite{aspremont2003relaxations,vandenberghe1996semidefinite}, with the standard form given by
\begin{equation}\label{eq:QCQP}
	\begin{aligned}
		\min_{z} \quad & z^\top P_0 z + q^\top_0 z + r_0\\
		\textrm{subject to} \quad & z^\top P_i z + q^\top_i z + r_i \leq 0, \quad i = 1, \ldots, m. \\
	\end{aligned}
\end{equation}

Although quadratic and linear terms do not appear simultaneously in our problem, we adopt the formulation in \eqref{eq:QCQP} from \cite{aspremont2003relaxations,vandenberghe1996semidefinite} for notational simplicity. Since martingale, marginal, and bicausal constraints enumerate all possible scenarios of $x$ and $y$, the number of constraints, denoted by $m$, can be very large. The non-convexity arises because $P_i, i = 1, \ldots, m,$ may not be positive semidefinite in our case. A standard convex relaxation method is known as the SDP relaxation using semidefinite programming. We include it here for the reader's convenience. The problem in \eqref{eq:QCQP} can be rewritten as
\begin{equation}\label{eq:QCQPtr}
	\begin{aligned}
		\min_{z, Z} \quad & \tr[Z P_0] + q^\top_0 z + r_0\\
		\textrm{subject to} \quad & \tr[Z P_i] + q^\top_i z + r_i \leq 0, \quad i = 1, \ldots, m, \\
		& Z = z z^\top. \\
	\end{aligned}
\end{equation}
We relax this non-convex problem into a convex one by replacing the constraint $Z = z z^\top$ with a (convex) positive semidefiniteness constraint $Z - z z^\top \succeq 0$, which is related to the Schur complement. Then the SDP relaxation of \eqref{eq:QCQP} is 
\begin{equation}\label{eq:SDP}
	\begin{aligned}
		\min_{z, Z} \quad & \tr[Z P_0] + q^\top_0 z + r_0\\
		\textrm{subject to} \quad & \tr[Z P_i] + q^\top_i z + r_i \leq 0, \quad i = 1, \ldots, m, \\
		& \begin{bmatrix}
			Z      & z \\
			z^\top & 1 \\
		\end{bmatrix} \succeq 0. \\
	\end{aligned}
\end{equation}

The optimal value of the SDP relaxation \eqref{eq:SDP} provides a lower bound for the optimal value of the non-convex problem in \eqref{eq:QCQP}.

An alternative approach to obtaining a lower bound for the non-convex problem \eqref{eq:QCQP} is to consider its dual, which is always convex. \cite{aspremont2003relaxations} derived the dual problem of \eqref{eq:QCQP} as follows:
\begin{equation}\label{eq:dual_QCQP}
	\begin{aligned}
		\max_{\lambda, \gamma} \quad & \gamma + \sum^m_{i=1} \lambda_i r_i + r_0\\
		\textrm{subject to} \quad & \begin{bmatrix}
			(P_0 + \sum^m_{i=1} \lambda_i P_i)  &  (q_0 + \sum^m_{i=1} \lambda_i q_i)/2 \\
			(q_0 + \sum^m_{i=1} \lambda_i q_i)^\top/2 & - \gamma \\
		\end{bmatrix} \succeq 0, \\
		& \lambda_i \geq 0, \quad i = 1, \ldots, m. \\
	\end{aligned}
\end{equation}
The problem in \eqref{eq:dual_QCQP} is a semidefinite program, known as the Lagrangian relaxation of \eqref{eq:QCQP}. Notably, the SDP relaxation \eqref{eq:SDP} and the Lagrangian relaxation \eqref{eq:dual_QCQP} are duals of each other and have the same optimal value if a constraint qualification holds \citet[Section 2.2]{aspremont2003relaxations}. Constraint qualifications are conditions under which the strong duality holds; see \citet[Section 5.2.3]{boyd2004convex} and Slater's condition for an example.

In general, strong duality between the non-convex QCQP \eqref{eq:QCQP} and the Lagrangian relaxation \eqref{eq:dual_QCQP} does not hold. However, \citet[Section 2.3]{aspremont2003relaxations} identifies a special case with a single constraint, where the duality gap is zero.

We summarize the standard results above in the following proposition and refer interested readers to \cite{aspremont2003relaxations,vandenberghe1996semidefinite} for more details. 

\begin{proposition}[\cite{aspremont2003relaxations,vandenberghe1996semidefinite}]
	When the supports $\cX$ and $\cY$ are discrete and finite, the bicausal MOT problem \eqref{eq:bicausal_discrete} is a special case of the non-convex QCQP \eqref{eq:QCQP}. A lower bound on the optimal value of \eqref{eq:QCQP} can be obtained using the SDP relaxation \eqref{eq:SDP} and the Lagrangian relaxation \eqref{eq:dual_QCQP}. These two relaxations are duals of each other and have the same optimal value if a constraint qualification holds.
\end{proposition}

As a non-convex QCQP, it is computationally challenging to find exact upper/lower bounds in the bicausal MOT problems \eqref{eq:bi_MOT}. In addition to SDP relaxations \citep{anstreicher2009semidefinite,sherali2002enhancing,aspremont2003relaxations,vandenberghe1996semidefinite}, other algorithms include branch-and-cut methods with reformulation-linearization techniques \citep{audet2000branch} and McCormick relaxations \citep{mccormick1976}. In this study, we choose the McCormick relaxation for several reasons. First, the bicausal MOT is motivated by option pricing, and the dual problem corresponds to the hedging problem, which is important in its own right. However, obtaining the dual formulation for the method in \cite{audet2000branch} is quite complex. In contrast, McCormick relaxations allow us to obtain the duality, even though the dual problem \eqref{prob:dual} remains lengthy. Second, compared to SDP relaxations, the McCormick MOT \eqref{eq:discrete} is linear, whereas the SDP relaxation \eqref{eq:SDP} remains nonlinear. Although the semidefinite programming is not much harder than LP to solve \citep{vandenberghe1996semidefinite}, the LP problem is still a bit more computable.

In practice, combining multiple optimization techniques may yield better bounds. However, in this paper, we focus on demonstrating the improvements achieved by a single technique. Additionally, McCormick relaxation is commonly used as a built-in routine in optimization software and is straightforward to implement manually.

\section{McCormick relaxation}\label{sec:McCormick}
\subsection{The finite discrete case}\label{sec:discrete}

McCormick relaxations construct the convex and concave envelopes as follows. For the probability $\pi(x_{1:N}, y_t)$, we suppose the upper bound $U_{N, t}(x_{1:N}, y_t)$ and lower bound $L_{N, t}(x_{1:N}, y_t)$ are known. The corresponding upper and lower bounds for $\pi(x_{1:t})$ are denoted as $U_{t, 0}(x_{1:t})$ and $L_{t, 0}(x_{1:t})$, respectively. Here, the subscripts in $L_{t, s}$ and $U_{t, s}$, $0 \leq t, s \leq N$, indicate that different functions are used when the variables $(x, y)$ have different indices. For simplicity, we denote $L_{t, s} =: L$ and $U_{t, s} =: U$, since the indices in $(x, y)$ imply the functions used. Then we can write  
\begin{align*}
	L(x_{1:N}, y_t) &\leq \pi(x_{1:N}, y_t) \leq U(x_{1:N}, y_t) \quad \text{ and }  \quad L(x_{1:t}) \leq \pi(x_{1:t}) \leq U(x_{1:t}).
\end{align*}
Since $(\pi(x_{1:N}, y_t) - L(x_{1:N}, y_t))(\pi(x_{1:t}) - L(x_{1:t})) \geq 0$, expanding the expression and substituting $w_1(x_{1:N}, y_t) = \pi(x_{1:N}, y_t) \pi(x_{1:t})$ yields
\begin{equation}\label{eq:prove_Mc}
	w_1(x_{1:N}, y_t) \geq L(x_{1:N}, y_t) \pi(x_{1:t}) + \pi(x_{1:N}, y_t) L(x_{1:t}) - L(x_{1:N}, y_t) L(x_{1:t}).
\end{equation}
Similarly, the same procedure applied to products involving $U(x_{1:N}, y_t) - \pi(x_{1:N}, y_t)$ and $U(x_{1:t}) - \pi(x_{1:t})$ yields three additional inequalities. Similarly, we also introduce $w_2(x_{1:N}, y_t) = \pi(x_{1:t}, y_t) \pi(x_{1:N})$, and four corresponding inequalities are derived. The causality constraint  \eqref{eq:bilin} becomes $w_1(x_{1:N}, y_t) = w_2(x_{1:N}, y_t)$. 

Together with the relaxation on anticausality and the martingale condition, we introduce McCormick martingale couplings as follows:
\begin{definition}\label{def:Mc}
	Given the upper bound $U(\cdot)$ and lower bound $L(\cdot)$ on probability masses, $\pi \in \Pi(\bar{\mu}, \bar{\nu})$ is called a McCormick martingale coupling if it satisfies
	\begin{enumerate}[label={(\arabic*)}]
		\item the capacity constraint: $L \leq \pi \leq U$;
		\item the martingale condition: 
		\begin{equation}
			\sum_{x_{t+1}} x_{t+1} \pi(x_{t+1} | x_{1:t}, y_{1:t}) = x_t, \quad \sum_{y_{t+1}} y_{t+1} \pi(y_{t+1} | x_{1:t}, y_{1:t}) = y_t, \quad 1 \leq t \leq N-1;
		\end{equation}
		\item the McCormick relaxation of the causality condition: for all $1 \leq t \leq N - 1$ and $x_{1:N}, y_t$, there exist variables $w_1(\cdot), w_2(\cdot) \in [0, 1]$ satisfying
		\begin{equation}\label{eq:Mc_cau}
			\begin{aligned}
				w_1(x_{1:N}, y_t) & \geq L(x_{1:N}, y_t) \pi(x_{1:t}) + \pi(x_{1:N}, y_t) L(x_{1:t}) - L(x_{1:N}, y_t) L(x_{1:t}), \\
				w_1(x_{1:N}, y_t) & \geq U(x_{1:N}, y_t) \pi(x_{1:t}) + \pi(x_{1:N}, y_t) U(x_{1:t}) - U(x_{1:N}, y_t) U(x_{1:t}),  \\
				w_1(x_{1:N}, y_t) & \leq U(x_{1:N}, y_t) \pi(x_{1:t}) + \pi(x_{1:N}, y_t) L(x_{1:t}) - U(x_{1:N}, y_t) L(x_{1:t}), \\
				w_1(x_{1:N}, y_t) & \leq \pi(x_{1:N}, y_t) U(x_{1:t}) + L(x_{1:N}, y_t) \pi(x_{1:t}) - L(x_{1:N}, y_t) U(x_{1:t}), \\
				w_2(x_{1:N}, y_t) & \geq L(x_{1:t}, y_t) \pi(x_{1:N}) + \pi(x_{1:t}, y_t) L(x_{1:N}) - L(x_{1:t}, y_t) L(x_{1:N}), \\
				w_2(x_{1:N}, y_t) & \geq U(x_{1:t}, y_t) \pi(x_{1:N}) + \pi(x_{1:t}, y_t) U(x_{1:N}) - U(x_{1:t}, y_t) U(x_{1:N}),  \\
				w_2(x_{1:N}, y_t) & \leq U(x_{1:t}, y_t) \pi(x_{1:N}) + \pi(x_{1:t}, y_t) L(x_{1:N}) - U(x_{1:t}, y_t) L(x_{1:N}), \\
				w_2(x_{1:N}, y_t) & \leq \pi(x_{1:t}, y_t) U(x_{1:N}) + L(x_{1:t}, y_t) \pi(x_{1:N}) - L(x_{1:t}, y_t) U(x_{1:N}), \\
				w_1(x_{1:N}, y_t)  & = w_2(x_{1:N}, y_t);
			\end{aligned}
		\end{equation}
		\item  the McCormick relaxation of the anticausality condition, by interchanging $x$ and $y$ in \eqref{eq:Mc_cau}.
	\end{enumerate}
	We denote the set of all McCormick martingale couplings as $\cM(\bar{\mu}, \bar{\nu}; L, U)$.
\end{definition}
All of the inequalities in \eqref{eq:Mc_cau} are derived in the same way as \eqref{eq:prove_Mc}. Note that \eqref{eq:Mc_cau} holds for all $x_{1:N}, y_t$, instead of merely $\pi$-a.s. In condition (3), it is possible to simplify $w_1(x_{1:N}, y_t) = w_2(x_{1:N}, y_t) =: w(x_{1:N}, y_t)$. The term $w(\cdot)$ (or $w_1(\cdot)$ and $w_2(\cdot)$) also serves as a variable in the optimization process. However, in \eqref{eq:Mc_cau}, we can eliminate the dependence on $w(\cdot)$ by imposing that the maximum of the lower bounds does not exceed the minimum of the upper bounds for $w(\cdot)$. Therefore, we only write $\pi \in \cM(\bar{\mu}, \bar{\nu}; L, U)$ to signify that $\pi$ is a McCormick martingale coupling.

A simple choice of bounds is $L = 0$ and $U(x_{1:N}, y_t) = \min\{\mu_1(x_1), \ldots, \mu_N(x_N), \nu_t(y_t)\}$, which are derived from the marginal conditions and do not require additional modeling assumptions. Under this choice, $\cM(\bar{\mu}, \bar{\nu}; L, U)$ is also not empty. Indeed, we can first construct (or rather arbitrarily pick) the marginal laws, $\mu\in\Pi(\overline\mu)$ and $\nu \in\Pi(\overline\nu)$, such that $X_{1:N}$ and $Y_{1:N}$ are martingales under $\mu$ and $\nu$, respectively. Then the independent coupling $\mu \otimes \nu$ serves as a McCormick martingale coupling. In the general finite and discrete cases, many LP solvers can check whether $\cM(\bar{\mu}, \bar{\nu}; L, U)$ is non-empty.  

\begin{remark}
	In Definition \ref{def:Mc}, the martingale condition is imposed under joint filtration, as the bicausal condition is relaxed and Proposition \ref{prop:joint_indiv} is not applicable in this context.
\end{remark}
\begin{remark}
	The total mass of $L$ or $U$ may not be equal to one. Both $L$ and $U$ are functions only, rather than probability measures or signed measures. Even when $U$ is a measure, the product $\pi (\cdot) U(\cdot)$ is not a measure. Specifically, for a set $A = B \cup C$ where $B$ and $C$ are disjoint, $\pi(A) U(A) = [\pi(B) + \pi(C)][U(B) + U(C)] \neq \pi(B)U(B) + \pi(C)U(C)$. Furthermore, it is crucial to interpret \eqref{eq:Mc_cau} pointwise in the discrete case, given that the McCormick relaxation is derived pointwise in \eqref{eq:prove_Mc}.
\end{remark}

We focus on addressing the minimization problem associated with the McCormick MOT in \eqref{eq:discrete}, as the maximization problem can be treated similarly:
\begin{equation}\label{eq:discrete}
	\inf_{\pi \in \cM(\bar{\mu}, \bar{\nu}; L, U)} \int c(x_{1:N}, y_{1:N}) \pi(dx_{1:N}, dy_{1:N}).
\end{equation} 
In finite and discrete scenarios, the infimum in the primal problem is attained when the set $\cM(\bar{\mu}, \bar{\nu}; L, U)$ is not empty. The dual problem is derived through classic finite-dimensional LP theory and is omitted here. The dual problem also has an optimizer when either the primal or dual optimal value is finite \cite[Proposition 5.2.1]{bertsekas2009convex}.

\subsection{The absolutely continuous case with capacity constraints}\label{sec:capacity}
In the continuous scenario similar to \cite{korman2015optimal, bogachev2022}, we make the assumption that the domains $\cX_t = \cY_t = \R$ and the marginals $\mu_t$ and $\nu_t$ are absolutely continuous with respect to the Lebesgue measure on $\R$. Similarly, we consider couplings $\pi$ that are absolutely continuous with respect to the Lebesgue measure $\lambda$ on $\R^N \times \R^N$. When the context is clear, we denote $\lambda(dx_{1:N}, dy_{1:N})$ by $d\lambda$ and omit the arguments. Let $f(x_{1:N}, y_{1:N})$ represent the density of a coupling $\pi$. Since $\pi$ is assumed to be absolutely continuous with respect to the Lebesgue measure $\lambda$, $\pi$ is a measure on the Lebesgue measurable sets. However, for probability measures, we typically work with the Borel measurable sets. Hence, we restrict to the Borel $\sigma$-algebra on the Euclidean space and interpret $\pi$ as a Borel probability measure. Additionally, nonnegative stock prices result in zero densities in negative domains.

Similarly as before, we denote the marginals of $f(x_{1:N}, y_{1:N})$ on $\cX_{1:t}$ as $f(x_{1:t})$, without adding extra subscripts to $f$. This convention applies to all other marginals as well. Due to absolute continuity, the causality constraint \eqref{eq:causal1} (or \eqref{eq:causal2}) transforms into the following density equality:
\begin{equation}\label{eq:conti}
	f(x_{1:N}, y_t) f(x_{1:t}) = f(x_{1:t}, y_t) f(x_{1:N}), \quad \text{$\lambda$-a.e.}, \quad t = 1, \ldots, N-1.
\end{equation}
Crucially, \eqref{eq:conti} holds $\lambda$-a.e., instead of merely $\pi$-a.s. It can be verified in the same spirit of \eqref{eq:bilin}, by noting that both sides of \eqref{eq:conti} are zero when regular conditional kernels are not defined. 

With a slight abuse of notation, let $w_1(x_{1:N}, y_t) = f(x_{1:N}, y_t) f(x_{1:t})$ still be the auxiliary variable. With the same convention on omitting the subscripts as in the discrete case, we suppose that the upper bound $u(\cdot)$ and lower bound $l(\cdot)$ for the density $f(\cdot)$ are known, such that 
\begin{align*}
	l(x_{1:N}, y_t) &\leq f(x_{1:N}, y_t) \leq u(x_{1:N}, y_t), \; \text{$\lambda$-a.e.} \quad \text{ and }  \quad l(x_{1:t}) \leq f(x_{1:t}) \leq u(x_{1:t}), \; \text{$\lambda$-a.e.},
\end{align*}
which are referred to as capacity/density constraints \cite{korman2015optimal,bogachev2022}. Similarly, the inequality \eqref{eq:prove_Mc} becomes
\begin{equation}\label{eq:prove_capa}
	w_1(x_{1:N}, y_t) \geq l(x_{1:N}, y_t) f(x_{1:t}) + f(x_{1:N}, y_t) l(x_{1:t}) - l(x_{1:N}, y_t) l(x_{1:t}), \quad \text{$\lambda$-a.e.}
\end{equation}
Other inequalities can be derived in a similar fashion. It is important to emphasize that $l(\cdot)$ and $u(\cdot)$ are not necessarily density functions.

Denote $g_t$ as the density of $\mu_t$, and $\bar{g} := (g_1, \ldots, g_N)$ as the individual density vectors. Similarly, let $h_t$ be the density of $\nu_t$, and $\bar{h} := (h_1, \ldots, h_N)$ be the individual density vectors. We define the continuous version of McCormick martingale couplings similarly. Note that \eqref{eq:Mc_capa} holds $\lambda$-a.e.
\begin{definition}\label{def:Mc_capa}
	Given the upper bound $u(\cdot)$ and lower bound $l(\cdot)$ on densities, a coupling $\pi$ with density $f \in \Pi(\bar{g}, \bar{h})$ is called a McCormick martingale coupling if it satisfies
	\begin{enumerate}[label={(\arabic*)}]
		\item the capacity constraint: $l \leq f \leq u$, $\lambda$-a.e.;
		\item the martingale condition: 
		\begin{equation}
			\begin{aligned}
				\int_{\cX_{t+1}} x_{t+1} f(x_{t+1} | x_{1:t}, y_{1:t}) \lambda(dx_{t+1}) &= x_t, \; 1 \leq t \leq N-1, \\ 
				\int_{\cY_{t+1}} y_{t+1} f(y_{t+1} | x_{1:t}, y_{1:t})\lambda(dy_{t+1}) &= y_t, \; 1 \leq t \leq N-1;
			\end{aligned}
		\end{equation}
		\item the McCormick relaxation of the causality condition on densities: for all $1 \leq t \leq N-1$, there exist $w_1(x_{1:N}, y_t), w_2(x_{1:N}, y_t) \geq 0$ such that the following inequality holds $\lambda$-almost everywhere:
		\begin{equation}\label{eq:Mc_capa}
			\begin{aligned}
				w_1(x_{1:N}, y_t) & \geq l(x_{1:N}, y_t) f(x_{1:t}) + f(x_{1:N}, y_t) l(x_{1:t}) - l(x_{1:N}, y_t) l(x_{1:t}), \\
				w_1(x_{1:N}, y_t) & \geq u(x_{1:N}, y_t) f(x_{1:t}) + f(x_{1:N}, y_t) u(x_{1:t}) - u(x_{1:N}, y_t) u(x_{1:t}),  \\
				w_1(x_{1:N}, y_t) & \leq u(x_{1:N}, y_t) f(x_{1:t}) + f(x_{1:N}, y_t) l(x_{1:t}) - u(x_{1:N}, y_t) l(x_{1:t}), \\
				w_1(x_{1:N}, y_t) & \leq f(x_{1:N}, y_t) u(x_{1:t}) + l(x_{1:N}, y_t) f(x_{1:t}) - l(x_{1:N}, y_t) u(x_{1:t}), \\
				w_2(x_{1:N}, y_t) & \geq l(x_{1:t}, y_t) f(x_{1:N}) + f(x_{1:t}, y_t) l(x_{1:N}) - l(x_{1:t}, y_t) l(x_{1:N}), \\
				w_2(x_{1:N}, y_t) & \geq u(x_{1:t}, y_t) f(x_{1:N}) + f(x_{1:t}, y_t) u(x_{1:N}) - u(x_{1:t}, y_t) u(x_{1:N}),  \\
				w_2(x_{1:N}, y_t) & \leq u(x_{1:t}, y_t) f(x_{1:N}) + f(x_{1:t}, y_t) l(x_{1:N}) - u(x_{1:t}, y_t) l(x_{1:N}), \\
				w_2(x_{1:N}, y_t) & \leq f(x_{1:t}, y_t) u(x_{1:N}) + l(x_{1:t}, y_t) f(x_{1:N}) - l(x_{1:t}, y_t) u(x_{1:N}), \\
				w_1(x_{1:N}, y_t) &= w_2(x_{1:N}, y_t);
			\end{aligned}
		\end{equation}
		\item  the McCormick relaxation of the anticausality condition on densities, by interchanging $x$ and $y$ in \eqref{eq:Mc_capa}.
	\end{enumerate}
	We use $f \in \cM(\bar{g}, \bar{h}; l, u)$ to denote that $f$ corresponds to a McCormick coupling with capacity constraints. 
\end{definition}

With continuous densities, we denote the risk-neutral price of the exotic option as
\begin{equation}\label{eq:pf}
	P(f) := \int c(x_{1:N}, y_{1:N}) f(x_{1:N}, y_{1:N}) d \lambda.
\end{equation}
The primal problem, incorporating McCormick relaxation and capacity constraints, can be expressed as
\begin{equation}\label{prob:capacity}
	P = \inf_{f \in \cM(\bar{g}, \bar{h}; l, u)} P(f).
\end{equation} 
This formulation is referred to as the continuous version of McCormick MOT.

Theorem \ref{thm:primal} demonstrates the primal attainment of the problem \eqref{prob:capacity}. The key step involves showing that $\cM(\bar{g}, \bar{h}; l, u)$ is compact under the weak topology of $L^1(\R^N \times \R^N)$. A natural question arises: does the primal attainment hold for the bicausal MOT in the absolutely continuous case? To establish that the corresponding feasible set is closed, density equalities such as \eqref{eq:conti} for causal or anticausal constraints require $\lambda$-a.e. convergence, at least for some subsequences. This issue also persists if we use the equivalent condition \eqref{prop:equi} instead. Since the weak topology on $L^p$ does not guarantee $\lambda$-a.e. convergence, proving primal attainment for the bicausal MOT in the absolutely continuous case remains an open problem.
\begin{theorem}\label{thm:primal}
	Suppose the following conditions hold:
	\begin{enumerate}[label={\arabic*.}]
		\item $\cM(\bar{g}, \bar{h}; l, u)$ is not empty;
		\item Probability densities $g_t, h_t \in L^1(\R)$ and have finite first moments. The bounds $l, u \in L^1(\R^N \times \R^N) \cap L^\infty(\R^N \times \R^N)$ and $l, u \geq 0$; 
		\item The cost (payoff) function $ c: \R^N \times \R^N \rightarrow (-\infty, \infty]$ is measurable and $c \geq - C(1 + \sum^N_{i=1}|x_i| + \sum^N_{j=1}|y_j|)$ for some constant $C \geq 0$.
	\end{enumerate}
	Then the infimum in \eqref{prob:capacity} is attained.
\end{theorem}

\subsubsection{The strong duality}\label{sec:dual}
It is well known that the dual formulation of the classic MOT corresponds to a semi-static hedging strategy, comprising the sum of static vanilla portfolios and dynamic delta positions \citep{beiglbock2013model}. A natural question arises regarding the impact of the McCormick relaxation on the hedging portfolio.  

In an informal manner, one can derive the dual problem by considering the Lagrange multipliers method and interchanging the infimum with the supremum. We introduce the multipliers as follows, having reduced $w_1(x_{1:N}, y_t) = w_2(x_{1:N}, y_t) = w(x_{1:N}, y_t)$: 
\begin{enumerate}[label={(\arabic*)}]
	\item $\phi_t$ and $\varphi_t$, $1 \leq t \leq N$ denote potential functions testing the marginal constraints for $\mu_t$ and $\nu_t$, respectively. In financial terms, they can be interpreted as a portfolio of static vanilla options;
	\item $\alpha_t$ and $\beta_t$, $1 \leq t \leq N-1$ test the martingale condition for $X$ and $Y$, respectively. They represent self-financing trading strategies in the risky assets;
	\item $\gamma_{t, i}$, $1 \leq t \leq N-1$, $1 \leq i \leq 8$, are multipliers for the McCormick relaxation of the causality condition; 
	\item $\eta_{t, i}$, $1 \leq t \leq N-1$, $1 \leq i \leq 8$ are multipliers for the McCormick relaxation of the anticausality condition; 
	\item $\kappa$ and $\theta$ are the multipliers for the upper and lower bounds, respectively.
\end{enumerate}
Appendix \ref{appen:derive_dual} provides an outline for deriving the dual problem. Constraints in the dual problem consist of the following components:
\begin{enumerate}[label={(\arabic*)}, leftmargin=*]
	\item The subhedging condition in the sense that
	\begin{equation}\label{simple_f_dual}
		\begin{aligned}
			& c(x_{1:N}, y_{1:N}) - \sum^N_{t=1} \phi_t(x_t) - \sum^N_{t=1}  \varphi_t(y_t) + \kappa(x_{1:N}, y_{1:N}) - \theta(x_{1:N}, y_{1:N})\\
			& + \sum^{N-1}_{t=1} \alpha_t(x_{1:t}, y_{1:t}) (x_{t+1} - x_t) + \sum^{N-1}_{t=1} \beta_t(x_{1:t}, y_{1:t}) (y_{t+1} - y_t) \\
			& + \Psi(x_{1:N}, y_{1:N}, \gamma, \eta; u, l) \geq 0,
		\end{aligned}
	\end{equation}
	where $\Psi(x_{1:N}, y_{1:N}, \gamma, \eta; u, l)$ is defined as the following sum:
	\begin{equation}\label{dual_f}
		\hspace{-\leftmargin}
		\begin{aligned}
			&\sum^{N-1}_{t=1} \gamma_{t, 1} (x_{1:N}, y_t) \big(l(x_{1:N}, y_t) + l(x_{1:t})\big) + \sum^{N-1}_{t=1} \gamma_{t, 2} (x_{1:N}, y_t) \big( u(x_{1:N}, y_t) + u(x_{1:t}) \big) \\
			& - \sum^{N-1}_{t=1} \gamma_{t, 3} (x_{1:N}, y_t) \big(u(x_{1:N}, y_t) + l(x_{1:t}) \big) - \sum^{N-1}_{t=1} \gamma_{t, 4} (x_{1:N}, y_t) \big( l(x_{1:N}, y_t)  + u(x_{1:t}) \big)  \\
			& + \sum^{N-1}_{t=1} \gamma_{t, 5} (x_{1:N}, y_t) \big( l(x_{1:t}, y_t) + l(x_{1:N}) \big) + \sum^{N-1}_{t=1} \gamma_{t, 6} (x_{1:N}, y_t) \big( u(x_{1:t}, y_t) + u(x_{1:N}) \big) \\
			& - \sum^{N-1}_{t=1} \gamma_{t, 7} (x_{1:N}, y_t) \big( u(x_{1:t}, y_t) + l(x_{1:N}) \big) - \sum^{N-1}_{t=1} \gamma_{t, 8} (x_{1:N}, y_t) \big( l(x_{1:t}, y_t) + u(x_{1:N}) \big) \\
			& + \sum^{N-1}_{t=1} \eta_{t, 1} (x_t, y_{1:N}) \big(l(x_t, y_{1:N}) + l(y_{1:t})\big) + \sum^{N-1}_{t=1} \eta_{t, 2} (x_t, y_{1:N}) \big( u(x_t, y_{1:N}) + u(y_{1:t}) \big) \\
			& - \sum^{N-1}_{t=1} \eta_{t, 3} (x_t, y_{1:N}) \big(u(x_t, y_{1:N}) + l(y_{1:t}) \big) - \sum^{N-1}_{t=1} \eta_{t, 4} (x_t, y_{1:N}) \big( l(x_t, y_{1:N})  + u(y_{1:t}) \big)  \\
			& + \sum^{N-1}_{t=1} \eta_{t, 5} (x_t, y_{1:N}) \big( l(x_t, y_{1:t}) + l(y_{1:N}) \big) + \sum^{N-1}_{t=1} \eta_{t, 6} (x_t, y_{1:N}) \big( u(x_t, y_{1:t}) + u(y_{1:N}) \big) \\
			& - \sum^{N-1}_{t=1} \eta_{t, 7} (x_t, y_{1:N}) \big( u(x_t, y_{1:t}) + l(y_{1:N}) \big) - \sum^{N-1}_{t=1} \eta_{t, 8} (x_t, y_{1:N}) \big( l(x_t, y_{1:t}) + u(y_{1:N}) \big);
		\end{aligned}
	\end{equation}
	\item The dual inequality for the McCormick relaxation of the causality condition:
	\begin{equation}\label{dual_Mc}
		\begin{aligned}
			& - \gamma_{t, 1}(x_{1:N}, y_t) - \gamma_{t, 2} (x_{1:N}, y_t) + \gamma_{t, 3}(x_{1:N}, y_t) + \gamma_{t, 4} (x_{1:N}, y_t) - \gamma_{t, 5}(x_{1:N}, y_t) \\
			&  - \gamma_{t, 6} (x_{1:N}, y_t) 	+ \gamma_{t, 7}(x_{1:N}, y_t) + \gamma_{t, 8} (x_{1:N}, y_t) \geq 0, \quad t = 1, \ldots, N-1.
		\end{aligned}
	\end{equation}
	
	\item The dual inequality for the McCormick relaxation of the anticausality condition:
	\begin{equation}\label{dual_AntiMc}
		\begin{aligned}
			& - \eta_{t, 1}(x_t, y_{1:N}) - \eta_{t, 2} (x_t, y_{1:N}) + \eta_{t, 3}(x_t, y_{1:N}) + \eta_{t, 4} (x_t, y_{1:N}) - \eta_{t, 5}(x_t, y_{1:N}) \\
			&  - \eta_{t, 6} (x_t, y_{1:N}) + \eta_{t, 7}(x_t, y_{1:N}) + \eta_{t, 8} (x_t, y_{1:N}) \geq 0, \quad t = 1, \ldots, N-1.
		\end{aligned}
	\end{equation}
\end{enumerate}

Define the set of feasible multipliers as
\begin{align}
	\Phi_c := \Big\{ \Big(\phi, \varphi, \alpha, \beta, \gamma, \eta, \kappa, \theta\Big)  \Big| & \phi_t, \varphi_t, \theta,  \kappa, \alpha_t, \beta_t, \gamma_{t, i},  \eta_{t, i}  \text{ are continuous and bounded}; \nonumber \\
	& \gamma_{t, i} \geq 0,  \eta_{t, i} \geq 0, \kappa \geq 0, \theta \geq 0; \nonumber \\
	& \text{\eqref{simple_f_dual}, \eqref{dual_Mc}, \eqref{dual_AntiMc} hold} \Big\}.
\end{align}

The objective function in the dual problem is
\begin{equation}\label{eq:dual_obj}
	\hspace{-\leftmargin}
	\begin{aligned}
		& D(\phi, \varphi, \alpha, \beta, \gamma, \eta, \kappa, \theta)\\
		& := \sum^N_{t=1} \int \phi_t(x_t) g_t(x_t) \lambda(dx_t) + \sum^N_{t=1} \int \varphi_t(y_t) h_t(y_t) \lambda(dy_t) \\
		& \quad - \int \kappa(x_{1:N}, y_{1:N}) u(x_{1:N}, y_{1:N}) d\lambda + \int \theta(x_{1:N}, y_{1:N}) l(x_{1:N}, y_{1:N}) d\lambda \\
		& \quad - \sum^{N-1}_{t=1} \int \Big[ \gamma_{t, 1} (x_{1:N}, y_t) l(x_{1:N}, y_t) l(x_{1:t}) + \gamma_{t, 2} (x_{1:N}, y_t) u(x_{1:N}, y_t) u(x_{1:t}) \Big] \lambda(dx_{1:N}, dy_t) \\
		& \quad + \sum^{N-1}_{t=1} \int \Big[ \gamma_{t, 3} (x_{1:N}, y_t) u(x_{1:N}, y_t) l(x_{1:t}) + \gamma_{t, 4} (x_{1:N}, y_t) l(x_{1:N}, y_t) u(x_{1:t}) \Big] \lambda(dx_{1:N}, dy_t) \\
		& \quad - \sum^{N-1}_{t=1} \int \Big[ \gamma_{t, 5} (x_{1:N}, y_t) l(x_{1:t}, y_t) l(x_{1:N}) + \gamma_{t, 6} (x_{1:N}, y_t) u(x_{1:t}, y_t) u(x_{1:N}) \Big] \lambda(dx_{1:N}, dy_t) \\
		& \quad + \sum^{N-1}_{t=1} \int \Big[ \gamma_{t, 7} (x_{1:N}, y_t) u(x_{1:t}, y_t) l(x_{1:N}) + \gamma_{t, 8} (x_{1:N}, y_t) l(x_{1:t}, y_t) u(x_{1:N}) \Big] \lambda(dx_{1:N}, dy_t) \\
		& \quad - \sum^{N-1}_{t=1} \int \Big[ \eta_{t, 1} (x_t, y_{1:N}) l(x_t, y_{1:N}) l(y_{1:t}) + \eta_{t, 2} (x_t, y_{1:N}) u(x_t, y_{1:N}) u(y_{1:t}) \Big] \lambda(dx_t, dy_{1:N}) \\
		& \quad + \sum^{N-1}_{t=1} \int \Big[ \eta_{t, 3} (x_t, y_{1:N}) u(x_t, y_{1:N}) l(y_{1:t}) + \eta_{t, 4} (x_t, y_{1:N}) l(x_t, y_{1:N}) u(y_{1:t}) \Big] \lambda(dx_t, dy_{1:N}) \\
		& \quad - \sum^{N-1}_{t=1} \int \Big[ \eta_{t, 5} (x_t, y_{1:N}) l(x_t, y_{1:t}) l(y_{1:N}) + \eta_{t, 6} (x_t, y_{1:N}) u(x_t, y_{1:t}) u(y_{1:N}) \Big] \lambda(dx_t, dy_{1:N}) \\
		& \quad + \sum^{N-1}_{t=1} \int \Big[ \eta_{t, 7} (x_t, y_{1:N}) u(x_t, y_{1:t}) l(y_{1:N}) + \eta_{t, 8} (x_t, y_{1:N}) l(x_t, y_{1:t}) u(y_{1:N}) \Big] \lambda(dx_t, dy_{1:N}).
	\end{aligned}
\end{equation}
Finally, we can state the dual problem as
\begin{equation}\label{prob:dual} 
	D = \sup_{\Phi_c}  D(\phi, \varphi, \alpha, \beta, \gamma, \eta, \kappa, \theta),
\end{equation}
where we omit the time subscript in $\phi, \varphi, \alpha, \beta$, etc., for simplicity.

Loosely speaking, McCormick relaxations modify the subhedging portfolio $(\phi, \varphi, \alpha, \beta)$ in \eqref{simple_f_dual} by considering
\begin{equation}
	c(x_{1:N}, y_{1:N}) + \kappa(x_{1:N}, y_{1:N}) - \theta(x_{1:N}, y_{1:N}) + \Psi(x_{1:N}, y_{1:N}, \gamma, \eta; u, l)
\end{equation}
in lieu of $c(x_{1:N}, y_{1:N})$. Additionally, the objective function \eqref{eq:dual_obj} is augmented with terms originating from capacity constraints and McCormick relaxations. However, elucidating these additional terms poses challenges due to the imposition of numerous inequalities in the McCormick envelopes.

We establish strong duality $P = D$, as defined in \eqref{prob:capacity} and \eqref{prob:dual}, through a two-step proof. First, we extend the capacity-constrained case presented in \citet[Theorem 1]{korman2015elementary} to encompass non-compact supports, where neither martingale conditions nor McCormick relaxations are imposed. Although our Lemma \ref{lem:capa} is acknowledged in the literature \citep{korman2015elementary}, a precise proof has not been provided in previous work. For later use, let $\Pi(g, h; u)$ represent the set of couplings with density $ 0 \leq f \leq u$ and marginals $g$ and $h$ on $\R$. While the proof focuses on one-dimensional cases for simplicity, the extension to multi-dimensional cases is analogous. The main idea involves a modification of \citet[Proposition 1.22]{villani2003topics}, specifically tailored to address the capacity constraint. Unlike Theorem \ref{thm:primal}, Lemma \ref{lem:capa} assumes lower semicontinuous (l.s.c.) costs. This condition is needed since the proof employs continuous and bounded functions for cost (payoff) approximation.
\begin{lemma}\label{lem:capa}
	Suppose the following conditions hold:
	\begin{enumerate}[label={\arabic*.}]
		\item The set of couplings $\Pi(g, h; u)$ is not empty;
		\item The upper bound $u$ satisfies $0 \leq u \in L^1(\R \times \R) \cap L^\infty(\R \times \R)$. Marginal probability densities $g, h \in L^\infty(\R)$ and have finite first moments;
		\item The cost $c: \R \times \R \rightarrow (-\infty, \infty]$ is l.s.c. and $c \geq - C(1 + |x| + |y|)$ for some constant $C \geq 0$.
	\end{enumerate}
	Then strong duality holds:
	\begin{equation}\label{eq:capa_dual}
		\inf_{f \in \Pi(g, h; u)} P(f)  = \sup_{\Psi_c} F(\phi, \varphi, \kappa),
	\end{equation}
	where $P(f)$ is defined in \eqref{eq:pf},
	\begin{equation*}
		F(\phi, \varphi, \kappa) := \int \phi(x) g(x) dx + \int \varphi(y) h(y) dy - \int \kappa(x, y) u(x, y) dxdy,
	\end{equation*}
	and
	\begin{align*}
		\Psi_c :=  \Big\{  (\phi, \varphi, \kappa) \Big| & \phi(x) + \varphi(y) \leq c(x, y) + \kappa(x, y), \, \kappa(x, y) \geq 0, \\
		& \kappa \in L^1(udxdy), \phi \in L^1(gdx), \varphi \in L^1(hdy) \Big\}.
	\end{align*}
	The infimum is also attained. Moreover, it does not change the value of the supremum on the right-hand side if we restrict $(\kappa, \phi, \varphi)$ to be continuous and bounded.
\end{lemma}

To establish the strong duality result using Lemma \ref{lem:capa}, we additionally assume that the upper and lower bounds are continuous and bounded in Theorem \ref{thm:dual}.
\begin{theorem}\label{thm:dual}
	Suppose the following conditions hold:
	\begin{enumerate}[label={\arabic*.}]
		\item $\cM(\bar{g}, \bar{h}; l, u)$ is not empty;
		\item The upper bound $u$ satisfies $0 \leq u \in L^1(\R^N \times \R^N) \cap  C_b(\R^N \times \R^N)$. The lower bound $l$ satisfies $0 \leq l \in L^1(\R^N \times \R^N) \cap C_b(\R^N \times \R^N)$. The probability densities $g_t, h_t \in L^\infty(\R)$ have finite first moments;
		\item The cost (payoff) function $ c: \R^N \times \R^N \rightarrow (-\infty, \infty]$ is l.s.c. and $c \geq - C(1 + \sum^N_{i=1}|x_i| + \sum^N_{j=1}|y_j|)$ for some constant $C \geq 0$.
	\end{enumerate}
	Then the strong duality holds: $P = D$.
\end{theorem}
As a remark on the dual attainment of McCormick MOT in the absolutely continuous case, we note that, even in the classic MOT setting, the corresponding set of dual variables lacks compactness. Without altering the formulation of the dual problem, dual attainment may fail in general. For a counterexample in the classic MOT setting, see \citet[Proposition 4.1]{beiglbock2016problem}. In contrast, for a one-dimensional, single time-step MOT problem, \cite{beiglbock2017complete} proved dual attainment in the quasi-sure sense, while \cite{beiglbock2019dual} recovered the pointwise dual attainment under additional assumptions on marginal distributions and/or cost function. \cite{lim2023MultiPeriodDuality} extended the results of \cite{beiglbock2017complete} to a multi-asset and multi-period setting. Furthermore, \cite{krvsek2024general} recently established dual attainment for the (bi-)causal OT problems. Although achieving dual attainment in general may be challenging, one could attempt to combine the techniques from \cite{beiglbock2017complete,lim2023MultiPeriodDuality,krvsek2024general} to prove the result (at least) in the quasi-sure sense, and under additional assumptions on the marginals (e.g., irreducibility). A natural question is whether Komlos-type arguments can be applied to our problem with complicated dual potentials. We leave this as future work.

\section{Numerical study}\label{sec:num}
In this section, we investigate the effectiveness of McCormick MOT using both synthetic and empirical data. Given the typical unavailability of risk-neutral densities, our empirical study includes the development of a calibration procedure capable of accommodating bid-ask spreads. We focus on basket and digital options.

\subsection{An illustrating example}
Assuming a zero risk-free rate for simplicity, we explore a scenario involving two stocks, $X$ and $Y$, covering two periods. An exotic option has a payoff given by $\max\{(X_2 - X_1)^2, (Y_2 - Y_1)^2 \}$. 

Setting the initial stock prices as $X_0 = 10$ and $Y_0 = 20$ for illustrative purposes, we posit that each stock can have only three different prices at each maturity. Specifically:
\begin{itemize}
	\item[(1)] At time $t = 1$, the stock price $X_1 = 11, 10, 9$ with probabilities $0.2, 0.6, 0.2$, respectively. $Y_1 = 24, 20, 16$ with probabilities $0.3, 0.4, 0.3$,  respectively;
	\item[(2)] At time $t = 2$, the stock price $X_2 = 20, 10, 0$ with probabilities $0.1, 0.8, 0.1$, respectively. $Y_2 = 26, 20, 14$ with probabilities $0.2, 0.6, 0.2$, respectively.
\end{itemize}
It is straightforward to solve this problem with optimization programs such as Gurobi. The classic MOT yields option price bounds of $[20.93, 24.40]$. To isolate the impact of McCormick relaxations, we utilize the lower and upper bounds implied by the marginal conditions: $L = 0$ and $U(x_{1:N}, y_t) = \min\{\mu_1(x_1), \ldots, \mu_N(x_N), \nu_t(y_t)\}$. Employing the McCormick MOT refines the bounds to $[21.50, 24.40]$. Despite the non-convex nature of bicausal MOT, (due to the small size of the problem) the program successfully addresses the problem by brute-force, providing bounds of $[21.64, 24.40]$. This simple example highlights the potential of McCormick MOT, prompting further exploration in more general cases.

\subsection{A calibration method of risk-neutral densities}
In the empirical case, McCormick MOT also needs risk-neutral distributions as input parameters. When calibrating $\mu_1$ and $\mu_2$ independently, there is no guarantee of maintaining convex order. To address this issue, we propose a methodology for calibrating the risk-neutral densities of a risky asset while maintaining the convex order between marginals. An alternative approach involves the separate calibration of $\mu_1$ and $\mu_2$, followed by convexification using the method outlined in \citet[Equation 3.1]{alfonsi2017sampling}. However, this approach often results in the modification of the support of $\mu_2$, raising uncertainties regarding the consistency of the modified $\mu_2$ with the options data.

Our construction is outlined as follows. The current time is denoted as $0$. At present, a finite number of European call options with known bid and ask prices are available. Recall that the set of maturities is represented as $T_1 < \ldots < T_t < \ldots < T_N$. For each maturity $T_t$, there are $n(t)$ options available, where $n(t)$ is not required to be the same across all maturities. Strike prices at a given maturity $T_t$ are sorted as $K_{1, t} < \ldots < K_{i, t} < \ldots < K_{n(t), t}$. In the context of considering the risk-neutral density of a single asset, we denote the underlying stock price at $T_t$ as $S_t := S_{T_t}$. Furthermore,
\begin{enumerate}[label={(\arabic*)}]
	\item $F_t$ represents the forward price of the stock, where the forward contract is initiated at time $0$ and delivered at time $T_t$;
	\item $D_t$ denotes the discount factor for time $T_t$, i.e., the price of a zero-coupon bond at time $0$ with maturity at time $T_t$;
	\item $A_{i, t}$ is the ask price of the call with the strike $K_{i, t}$ and maturity $T_t$. Similarly, $B_{i, t}$ denotes the bid price.  
\end{enumerate}
We refer to $C_{i, t}$ as a feasible price if $C_{i, t} \in [B_{i, t}, A_{i, t}]$. To streamline the framework, we assume that the interest rate is deterministic and that any dividends paid by the underlying assets are also deterministic if applicable. Following the approach of \cite{davis2007range}, we introduce scaled prices denoted as
\begin{equation}
	c_{i, t} = \frac{C_{i, t}}{D_t F_t}, \quad a_{i, t} = \frac{A_{i, t}}{D_t F_t}, \quad b_{i, t} = \frac{B_{i, t}}{D_t F_t}, \quad k_{i, t} = \frac{K_{i, t}}{F_t}.
\end{equation}

We write $\mu(S_{1:N} = s_{1:N})$ as a joint risk-neutral measure encompassing all maturities. For simplicity, we denote $\mu_t(S_t = s_t)$ as the marginal of $\mu$ on $S_t$. Similarly, $\mu(S_{1:t} = s_{1:t})$ refers to the marginal of $\mu$ on $S_{1:t}$. Suppose that the support of the stock price at time $T_t$, denoted as $\cS_t$, is given by the set of possible strikes at time $T_t$, i.e., $\cS_t = \{K_{1, t}, \ldots, K_{n(t), t}\}$. Consequently, $\mu$ is an $n(1) \times n(2) \times \ldots \times n(T)$ array.

To find the risk-neutral measure $\mu$, we consider the following optimization problem:
{\allowdisplaybreaks
	\begin{align}
		& \min_{c_{i, t}, \; \mu}  \sum_{\substack{1 \leq t \leq N \\ 1 \leq i \leq n(t)} } |c_{i, t} - a_{i, t}| + | c_{i, t} - b_{i, t}| \label{eq:mu_obj}\\
		& \textrm{ subject to} \sum_{s_t \in \cS_t} \left( \frac{s_t}{F_t} - k_{i, t}\right)^+ \mu_t(S_t = s_t) = c_{i, t}, \label{eq:price}\\
		& \qquad \qquad \sum_{s_{t+1} \in \cS_{t+1}} \frac{s_{t+1}}{F_{t+1}}\mu(s_{1:t}, s_{t+1}) = \frac{s_t}{F_t} \mu(S_{1:t} = s_{1:t}), \quad \forall \; s_{1:t}, \; 1 \leq t \leq N-1, \label{eq:martingale}\\
		& \qquad \qquad \sum_{s_t \in \cS_t} s_t \mu_t(S_t=s_t) = F_t, \quad \forall \; 1 \leq t \leq N, \label{eq:mean}\\
		& \qquad \qquad \mu(S_{1:N} = s_{1:N}) \geq 0, \quad \sum_{s_t \in \cS_t} \mu_t(S_t = s_t) = 1, \quad \forall \; 1 \leq t \leq N. \label{eq:prob}
	\end{align}%
}
The objective function \eqref{eq:mu_obj} aims to place the feasible price $c_{i, t} \in [b_{i, t}, a_{i, t}]$. The constraint \eqref{eq:price} represents the risk-neutral pricing formula, while \eqref{eq:martingale} enforces the martingale condition. \eqref{eq:mean} stems from the fact that the expected value of stock prices under a risk-neutral measure equals the forward price. The final condition ensures that $\mu_t$, the marginal of $\mu$ at time $T_t$, acts as a probability mass function. A notable advantage of our framework is that the program \eqref{eq:mu_obj} is still an LP problem, although the solution is not necessarily unique. When there is a risk-neutral measure $\mu$ that guarantees $c_{i, t} \in [b_{i, t}, a_{i, t}]$, the objective value is the lowest and equals $\sum_{i, t} |a_{i, t} - b_{i, t}|$. When the optimized objective value exceeds $\sum_{i, t} |a_{i, t} - b_{i, t}|$, it indicates the existence of a risk-neutral measure $\mu$, but certain options lack a feasible price $c_{i, t} \in [b_{i, t}, a_{i, t}]$ under $\mu$. Although only marginals $(\mu_1, \ldots, \mu_N)$ are required, our program in \eqref{eq:mu_obj} actually obtains a joint martingale probability distribution. Consequently, it ensures the convex order of marginals automatically. Besides, it addresses bid-ask spreads without solely relying on the mid-price, i.e., the average of bid and ask prices.

We summarize the above discussion in the following proposition.
\begin{proposition}\label{prop:cali}
	The program \eqref{eq:mu_obj}-\eqref{eq:prob} is an LP problem. If the optimal value equals to the lowest possible value, $\sum_{i, t} |a_{i, t} - b_{i, t}|$, then there exists a risk-neutral measure $\mu$, and each option has a (scaled) price $c_{i, t} \in [b_{i, t}, a_{i, t}]$ under $\mu$. If the optimal value is finite but greater than $\sum_{i, t} |a_{i, t} - b_{i, t}|$, then the risk-neutral measure with all (scaled) prices $c_{i, t} \in [b_{i, t}, a_{i, t}]$ does not exist. Instead, there exists a risk-neutral measure $\mu$, with at least one (scaled) price $c_{i, t} \notin [b_{i, t}, a_{i, t}]$ under $\mu$.
\end{proposition}

\subsection{Basket options}\label{sec:basket}
\subsubsection{Equal weights}\label{sec:equal}
The effectiveness of McCormick relaxations depends on the characteristics of the option payoff. Since causality imposes conditional independence, we find that McCormick relaxations are more effective when the payoff function is path-dependent for both assets. Consider a basket call option with an Asian-style payoff, averaged across two stocks $(X, Y)$ and two periods:
\begin{equation}\label{eq:payoff}
	\max\left\{\frac{X_1+ X_2 + Y_1 + Y_2}{4} - K, 0 \right\}.
\end{equation}
First, we employ probability bounds implied by marginal constraints. Although these bounds are typically loose, the McCormick MOT can still yield tighter price intervals compared to the classic MOT framework.

We obtained the zero-coupon yield curve, forward prices, and European option prices from OptionMetrics via Wharton Research Data Services (WRDS). In the first example, consider an investor interested in stocks from two pharmaceutical companies, namely Gilead Sciences (Ticker: GILD) and GSK plc (Ticker: GSK). According to OptionMetrics data, the liquidity of European options for these stocks is considered moderate. For example, on 28 February 2023, GSK options show an open interest of 4458 and a volume of 167401, while GILD options have an open interest of 4989 and a volume of 233003. We utilize Gurobi to solve the linear program \eqref{eq:mu_obj} and derive the risk-neutral marginals $(\mu_1, \mu_2)$ and $(\nu_1, \nu_2)$. To assess the reduction in basket option price bounds, we introduce the following ratio:
\begin{equation}\label{eq:ratio}
	\frac{\text{McCormick Max} - \text{McCormick Min}}{ \text{MOT Max} - \text{MOT Min}}.
\end{equation}

For the tenors, we select $T_1$ as the maturity date closest to the current time $T_0 = 0$. The second maturity, $T_2$, is approximately one month after $T_1$. In the context of the payoff \eqref{eq:payoff}, the strike $K$ is determined as the integer rounded off from the average of forward prices with delivery dates $T_1$ and $T_2$. Consequently, the basket option is close to the at-the-money level. We employ Gurobi to solve both the classic and McCormick MOT, treating them as LP problems. For the computational complexity, our problems typically have around 10,000 variables and 1,000 linear inequality constraints, which can be efficiently solved in less than a second on a laptop. Table \ref{tab:GILD} provides examples of price limits for the basket option and the ratios defined in \eqref{eq:ratio}. Considering $T_0$ ranging from February 28, 2022, to February 28, 2023 (the most recent one-year horizon in WRDS), we derive ratio values and present a histogram in Figure \ref{fig:GILD}. The average ratio is approximately 96.10\%, which signifies a 3.90\% reduction in price bounds. In the best-case scenario, the bounds are reduced by 20\%. Our code and Excel files detailing the ratios are accessible at \url{https://github.com/hanbingyan/McCormick}.

When considering stocks with a more liquid option market, the reduction in the price gap is less pronounced. As an illustration, on February 28, 2023, European options for JPMorgan Chase (JPM) reported an open interest of 54,034 and a volume of 1,189,048, while European options for Morgan Stanley (MS) exhibited an open interest of 22,885 and a volume of 669,638. The corresponding results are presented in Table \ref{tab:JPM}, and Figure \ref{fig:JPM} provides the histogram of the ratios. On average, the application of the McCormick MOT results in a 1.08\% reduction in the bounds. This outcome can be attributed to larger supports of risk-neutral densities for JPM and MS.

\begin{table}
	\tiny
	\centering
	\begin{tabular}{ccccccccc}
		\hline
		$T_0$ &    $T_1$ &  $T_2$  &  Strike &  MOT Max &  MOT Min &  McCormick Max &  McCormick Min &  Ratio \\
		\hline
		2022-02-28 & 2022-03-04 & 2022-04-01 & 50 & 1.8727 & 1.3967 & 1.8709 & 1.4041 & 0.9805 \\
		2022-04-26 & 2022-04-29 & 2022-05-27 & 52 & 1.9539 & 1.3820 & 1.9528 & 1.3840 & 0.9947 \\
		2022-06-23 & 2022-06-24 & 2022-07-22 & 52 & 1.3084 & 0.9185 & 1.2730 & 0.9209 & 0.9032 \\
		2022-08-26 & 2022-09-02 & 2022-09-30 & 47 & 1.1897 & 0.8783 & 1.1718 & 0.8978 & 0.8796 \\
		2022-10-24 & 2022-10-28 & 2022-11-25 & 50 & 1.4890 & 0.6643 & 1.4836 & 0.6684 & 0.9885 \\
		2022-12-20 & 2022-12-23 & 2023-01-20 & 59 & 1.6527 & 1.1156 & 1.6197 & 1.1222 & 0.9265 \\
		2023-02-17 & 2023-02-24 & 2023-03-24 & 60 & 1.2341 & 0.7453 & 1.2311 & 0.7566 & 0.9708 \\
		\hline
	\end{tabular}
	\caption{Basket options on GILD and GSK. Time is in the Year-Month-Day format. }\label{tab:GILD}
\end{table}

\begin{table}
	\tiny
	\centering
	\begin{tabular}{ccccccccc}
		\hline
		$T_0$ &    $T_1$ &  $T_2$  &  Strike &  MOT Max &  MOT Min &  McCormick Max &  McCormick Min &  Ratio \\
		\hline
		2022-02-28 & 2022-03-04 & 2022-04-01 & 116 & 3.5667 & 1.7953 & 3.5637 & 1.8074 & 0.9914 \\
		2022-04-26 & 2022-04-29 & 2022-05-27 & 102 & 3.3090 & 1.3595 & 3.3036 & 1.3735 & 0.9901 \\
		2022-06-23 & 2022-06-24 & 2022-07-22 & 93 & 2.8894 & 1.2243 & 2.8822 & 1.2314 & 0.9914 \\
		2022-08-19 & 2022-08-26 & 2022-09-23 & 104 & 2.5952 & 0.7202 & 2.5733 & 0.7402 & 0.9777 \\
		2022-10-17 & 2022-10-21 & 2022-11-18 & 96 & 3.0706 & 0.8126 & 3.0667 & 0.8311 & 0.9901 \\
		2022-12-13 & 2022-12-16 & 2023-01-13 & 113 & 2.8779 & 0.7043 & 2.8770 & 0.7269 & 0.9892 \\
		2023-02-10 & 2023-02-17 & 2023-03-17 & 120 & 2.5703 & 0.3847 & 2.5640 & 0.3964 & 0.9918 \\
		\hline
	\end{tabular}
	\caption{Basket options on JPM and MS.}\label{tab:JPM}
\end{table}

%

\begin{figure}
	\centering
	\begin{minipage}{0.45\textwidth}
		\centering
		\includegraphics[width=0.95\textwidth]{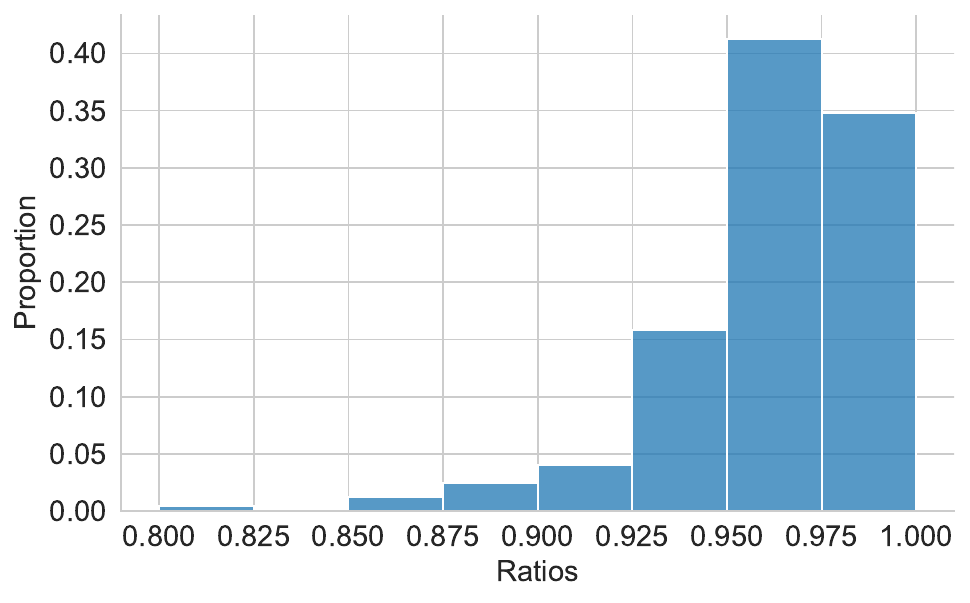}
		\subcaption{GILD and GSK as underlying assets}\label{fig:GILD}
	\end{minipage}
	\begin{minipage}{0.45\textwidth}
		\centering
		\includegraphics[width=0.95\textwidth]{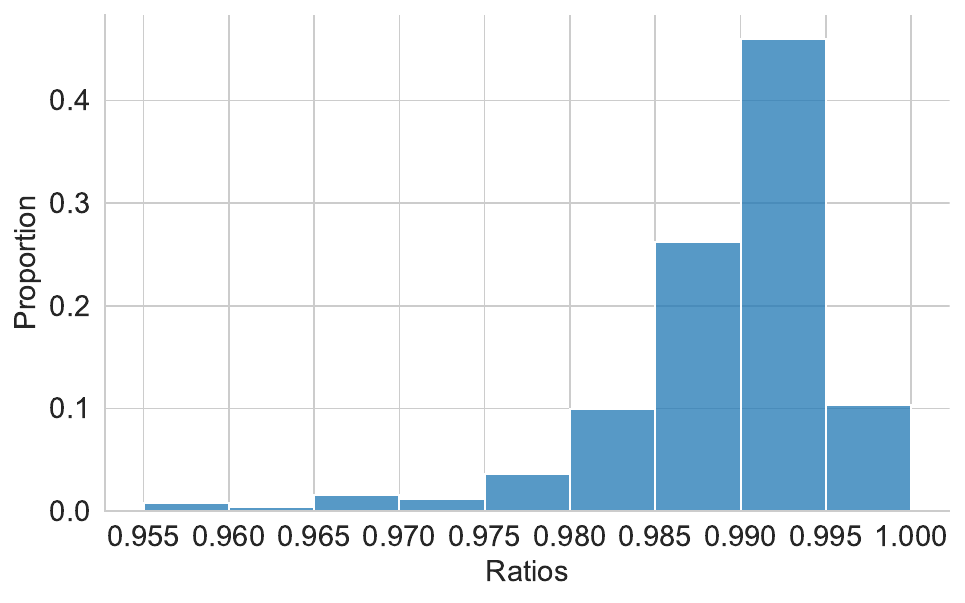}
		\subcaption{JPM and MS as underlying assets}\label{fig:JPM}
	\end{minipage}%
	\caption{Distribution of price ratios.}
\end{figure}

\subsubsection{Tighter bounds on probability masses}\label{sec:tighter_bnds}
In previous instances, the natural lower and upper bounds $L$ and $U$ on probability masses may have been relatively loose, necessitating tighter constraints using additional information. Our methodology integrates seamlessly with other knowledge available to the agent. For instance, actively traded exotic options, such as digital options, can provide upper or lower bounds on probability densities due to their specific payoffs. Although we do not have access to such proprietary data, we demonstrate the effectiveness of our method using artificial bounds. Readers can easily apply the same procedure once relevant data become available. In this subsection, we focus on the basket option \eqref{eq:payoff} on GILD and GSK. Suppose the agent imposes the following constraints on the coupling $\pi$:
\begin{align}
	\pi(x_{1:N}, y_{1:N}) & \leq 0.01, \label{eq:upperbd} \\
	\pi(x_{1:N}, y_{1:N}) & \geq \zeta \min\{\mu_1(x_1), \ldots, \mu_N(x_N), \nu_1(y_1), \ldots, \nu_N(y_N)\}, \text{ for } (x_{1:N}, y_{1:N}) \in \cA. \label{eq:lowerbd}
\end{align}
Here, $\zeta > 0$ serves as a shrinking parameter, and $\cA$ is a subset of paths where \eqref{eq:lowerbd} holds. Empirically, there are typically more than 2,000 paths available for $(x_{1:N}, y_{1:N})$. The upper bound of 0.01 in \eqref{eq:upperbd} ensures that the coupling $\pi$ does not concentrate excessively on too few paths; this value is chosen arbitrarily for demonstration purposes. For the lower bound \eqref{eq:lowerbd}, the agent specifies that certain paths of interest should have a sufficiently large probability. To make \eqref{eq:upperbd} and \eqref{eq:lowerbd} feasible, we can set $\zeta \leq 0.01$. Lacking prior knowledge of the underlying assets GILD and GSK, the set $\cA$ is chosen arbitrarily for demonstration. For each support $\cX_1$, $\cX_2$, $\cY_1$, and $\cY_2$, we select one-third of the possible stock prices in an equally spaced manner, resulting in the set $\cA$ containing at most $1/81 \approx 1.23\%$ of the paths. The set of McCormick martingale couplings $\cM(\bar{\mu}, \bar{\nu}; L, U)$ can be empty after imposing \eqref{eq:upperbd} and \eqref{eq:lowerbd}. However, the existence of a McCormick martingale coupling can be numerically checked using LP solvers such as Gurobi. 

Table \ref{tab:tighter} presents the statistics of the ratios obtained under different shrinking parameters $\zeta$ and maturities $T_2$. We consider $\zeta \in \{ 0.005, 0.01 \}$ for simplicity. The second maturity date $T_2$ can be 1, 2, 3 or 4 weeks after the first maturity date $T_1$, where $T_1$ is the maturity date closest to the current time $T_0 = 0$. Overall, McCormick MOT yields much better price intervals when tighter bounds $L$ and $U$ are available. On average, the eight cases in Table \ref{tab:tighter} achieve a reduction of 12.26\%. The best instance, with a minimum ratio of 0.1405, indicates a reduction of more than 80.00\%. Notably, the reductions here are relative to the classic MOT method with the same tighter constraints $L$ and $U$. Table \ref{tab:tighter} demonstrate that McCormick relaxations can further enhance price intervals, primarily because the inequalities like \eqref{eq:Mc_cau} become more effective. Besides, the performance of McCormick relaxations is better when the maturity $T_2$ is shorter since the number of paths is smaller, and the upper bound \eqref{eq:upperbd} is more likely to be binding. Generally, a larger parameter $\zeta$ results in greater improvements in price bounds. However, this is not always the case, as the set $\cM(\bar{\mu}, \bar{\nu}; L, U)$ is more likely to be empty with larger $\zeta$.

\begin{table}[h]
	\small
	\centering
	\begin{tabular}{ccccc}
		\hline
		$\zeta$ & Index for $T_2$ & Ratios Max & Ratios Mean & Ratios Min \\
		\hline		
		0.005 & 1 & 0.9827 & 0.8464 & 0.5655 \\
		0.005 & 2 & 0.9709 & 0.8807 & 0.7407 \\
		0.005 & 3 & 0.9891 & 0.8908 & 0.2829 \\
		0.005 & 4 & 0.9882 & 0.8898 & 0.2247 \\
		0.01 & 1 & 0.9802 & 0.8806 & 0.7013 \\
		0.01 & 2 & 0.9627 & 0.8677 & 0.5256 \\
		0.01 & 3 & 0.9835 & 0.8505 & 0.1405 \\
		0.01 & 4 & 0.985 & 0.9125 & 0.5726 \\
		\hline
	\end{tabular}
	\caption{Basket option on GILD and GSK with tighter bounds \eqref{eq:upperbd} and \eqref{eq:lowerbd}. }\label{tab:tighter}
\end{table}

Although \cite{eckstein2021robust} used a different approach by incorporating additional maturities, a comparative analysis of reductions with their findings provides information on the significance of our results. In \citet[Table 4.3]{eckstein2021robust}, focusing on uniform marginals and spread options, the reported price intervals are $[0.335, 8.337]$ for a single period, $[0.416, 8.254]$ for two periods and $[0.776, 7.920]$ for four periods. Consequently, in their four-period scenario, the calculated ratio is $(7.920 - 0.776)/(8.337-0.335) = 89.28\%$, falling within the range observed in our Table \ref{tab:tighter}.

\subsubsection{Arbitrary weights}\label{sec:arbitrary_w}
To further assess the performance of McCormick MOT on basket options, we generalize the weight parameter for stock prices and define the option payoff as
\begin{equation}
	\max\left\{ w_1 X_1 + w_2 X_2 + w_3 Y_1 + w_4 Y_2 - K, 0 \right\},
\end{equation}
where $w_1 + w_2 + w_3 + w_4 = 1$ and $w_i \geq 0$. Unlike the previous subsection, we only use the natural upper and lower bounds implied by marginal constraints.

We consider 100 arbitrarily chosen weight tuples $w := (w_1, w_2, w_3, w_4)$. For each tuple, we test it on the underlying assets GILD and GSK, which have 247 distinct initial dates for the option. The strike price is set as the average of the forward prices for $X_1$ and $Y_1$, delivered at $T_1$. After calculating the ratios for each $w$ on each date, we obtain a $247 \times 100$ array of ratios, and present the statistics as follows. 

\begin{figure}
	\centering
	\begin{minipage}{0.45\textwidth}
		\centering
		\includegraphics[width=0.95\textwidth]{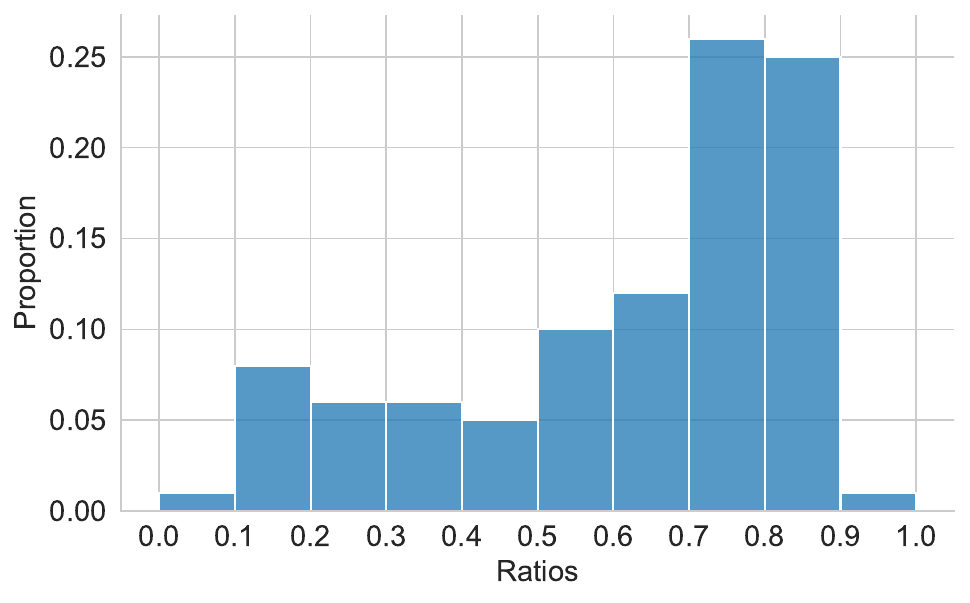}
		\subcaption{Minimum ratio achievable for a given weight $w$ across 247 dates.}\label{min_time}
	\end{minipage}
	\begin{minipage}{0.45\textwidth}
		\centering
		\includegraphics[width=0.95\textwidth]{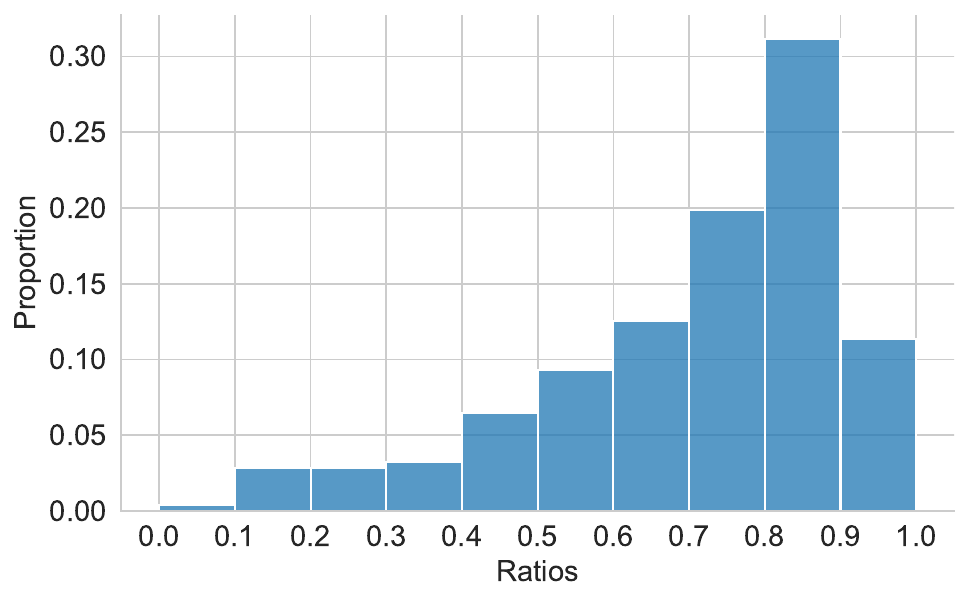}
		\subcaption{Minimum ratio achievable on a given date across 100 weight parameters.}\label{min_weight}
	\end{minipage}
	\begin{minipage}{0.45\textwidth}
		\centering
		\includegraphics[width=0.95\textwidth]{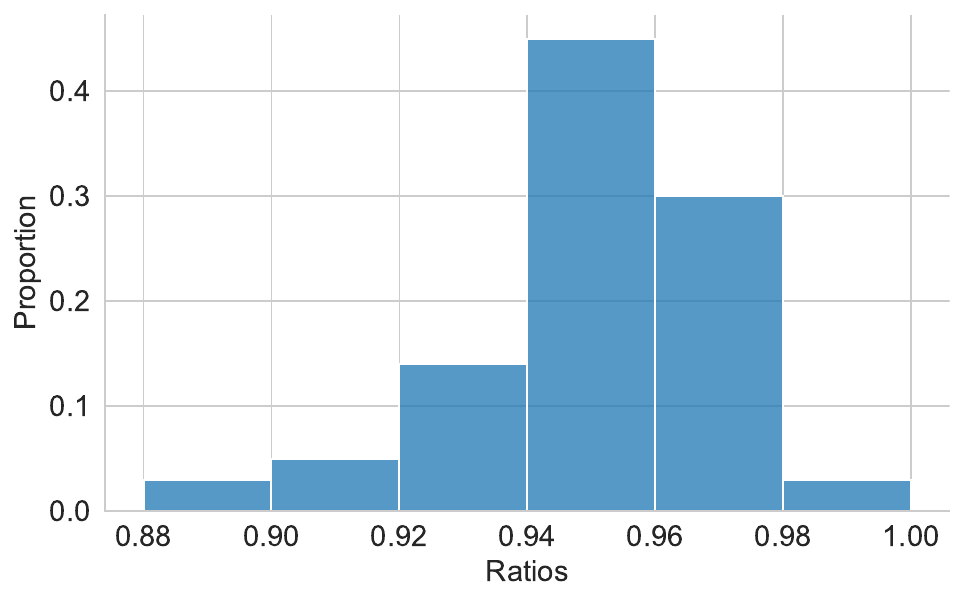}
		\subcaption{Average ratios across all dates with fixed weight parameters.}\label{mean_time}
	\end{minipage}%
	\caption{Distribution of price ratios.}
\end{figure}

\begin{enumerate}[label=(\arabic*)]
	\item For each given weight $w$, Figure \ref{min_time} shows the minimum ratio achievable across all dates. For most weight parameters, there is at least one date where the ratio falls below $90.00\%$. Approximately half of the weight parameters result in a minimum ratio between $70.00\%$ and $90.00\%$. 
	\item Similarly, we can analyze the data from another perspective. For each given date, Figure \ref{min_weight} shows the minimum ratio achievable across 100 possible weights. For over $85.00\%$ of the $247$ days considered, there exists at least one weight that yields a ratio lower than $0.9$.
	\item For each weight parameter, we calculate the average ratio across all dates. Figure \ref{mean_time} presents the histogram of these average ratios. Most weight parameters yield an average ratio between $94.00\%$ and $98.00\%$, consistent with the equal-weight case. However, there is a weight tuple that can achieve an average ratio of approximately $88.40\%$.
\end{enumerate}

%
%

\subsection{Digital options}\label{sec:digital}
Motivated by digital options, we present another example where the option payoffs are represented by indicator functions at atoms. In this case, the expectation corresponds to the probability at specific atom points. For numerical stability, we also set the notional amount as $\$10,000$. Then the expected payoff is given by
\begin{equation}\label{eq:digital_payoff}
	10000 \int c(x_{1:2}, y_{1:2}) \pi(dx_{1:2}, dy_{1:2}) = 10000 \times \pi(X_1 = x_1, X_2 = x_2, Y_1 = y_1, Y_2 = y_2).
\end{equation} 
In all cases for digital options, only natural bounds implied by marginal constraints are used.

\subsubsection{An illustrating example}

Suppose the initial stock prices are $X_0 = 2$ and $Y_0 = 3$, respectively.
\begin{itemize}
	\item[(1)] At time $t = 1$, the stock price $X_1 = 1, 2, 3$ with probabilities $0.01, 0.98, 0.01$, respectively. $Y_1 = 2, 3, 4$ with probabilities $0.4, 0.2, 0.4$,  respectively;
	\item[(2)] At time $t = 2$, the stock price $X_2 = 1, 2, 3$ with probabilities $0.04, 0.92, 0.04$, respectively. $Y_2 = 2, 3, 4$ with probabilities $0.4, 0.2, 0.4$, respectively.
\end{itemize}
Here, the asset $X$ exhibits very low volatility, while the asset $Y$ shows very high volatility.

Consider the payoff $c = 10000 \times \one_{\{X_1 = x_1, X_2 = x_2, Y_1 = y_1, Y_2 = y_2 \}}$. There are $ 3\times 3 \times 3 \times 3 = 81$ combinations of $(x_1, x_2, y_1, y_2)$. Given the small problem size, the bicausal MOT can be solved using a brute-force approach. We observe the following pattern: Bicausal or McCormick MOT may not improve prices, resulting in a ratio of 1.0. Only nine out of the 81 cases show price reductions. However, when the proposed methods do improve prices, the reductions are significant. The best case is shown in Table \ref{tab:toy_digital} for the event $\{X_1 = 2, X_2 = 3, Y_1 = 3, Y_2 = 3\}$, where both bicausal and McCormick MOT lead to the same price. In this case, the ratio is $0.0204$, indicating a $97.96\%$ reduction in the price gap.

The average ratio across all $81$ cases is $89.91\%$. However, the mean may not be the best metric here. The ratio is either very close to $0$ or exactly $1$. Therefore, we also report median values in subsequent examples.

Another observation is that bicausal and McCormick MOT yield the same prices in all $81$ cases, highlighting the effectiveness of McCormick relaxations.

\begin{table}[H]
	\centering
	\begin{tabular}{ccccc}
		\hline
		MOT Max &  MOT Min &  McCormick Max &  McCormick Min &  Ratio \\
		\hline
		300.0 &  0.0 &    61.2245 &  55.1020 & 0.0204 \\
		\hline
	\end{tabular}
	\caption{The best case in the illustrative example. Bicausal MOT achieves identical prices to McCormick MOT and is therefore omitted.}\label{tab:toy_digital}
\end{table}

In general, we can also consider payoffs depending on more general events, such as $\{ X_1 \leq 2, X_2 = 3, Y_2 \geq 3 \}$. The corresponding results are presented in Table \ref{tab:toy_events}. In this case, both the bicausal and McCormick MOT yield the same prices.

\begin{table}[H]
	\centering
	\begin{tabular}{ccccc}
		\hline
		MOT Max &  MOT Min &  McCormick Max &  McCormick Min &  Ratio \\
		\hline
		300.0 &  0.0      &    183.6735 &  177.5510 & 0.0204 \\
		\hline
	\end{tabular}
	\caption{Prices for the payoff as the indicator function of the event $\{ X_1 \leq 2, X_2 = 3, Y_2 \geq 3 \}$.}\label{tab:toy_events}
\end{table}

\subsubsection{Empirical examples}

In Section \ref{sec:basket}, the marginals calibrated from empirical data result in thousands of possible cases, making it impractical to solve the bicausal MOT in a reasonable time. Therefore, this section only focuses on McCormick MOT for digital options. Notably, it can significantly reduce prices even with liquid options on assets such as JPM and MS. 

As in Section \ref{sec:basket}, for each initial date $T_0$, the first maturity date $T_1$ is the one closest to $T_0$, and the second maturity date $T_2$ is roughly one month after $T_1$. Payoffs are indicator functions of all possible atom points $\{X_1 = x_1, X_2 = x_2, Y_1 = y_1, Y_2 = y_2\}$ after fixing the dates. Due to the large number of combinations, Table \ref{tab:three} presents results only for three arbitrarily chosen initial dates $T_0$. A similar pattern to the illustrative example is observed: The mean value of the ratio is significantly low, while the median is higher since prices are not reduced in about half of the cases. 

\begin{table}[H]
	\centering
	\begin{tabular}{ccccc}
		\hline
		Initial Date & Marginals Size & Ratios Mean &  Ratios Median & Ratios Min \\
		\hline
		02-28-2022 & $10 \times 9 \times 10 \times 8$	& 0.7885 & 1.0 & 0.0466  \\
		05-24-2022 & $11 \times 12 \times 8 \times 9$ & 0.7536 & 0.9118 & 0.0158  \\
		08-25-2022 & $10 \times 11 \times 7 \times 12$ & 0.6685 & 0.7925 & 0.0061 \\
		\hline
	\end{tabular}
	\caption{Price ratios for three arbitrary initial dates considering all possible atom points.}\label{tab:three}
\end{table}

The best cases for each initial date are presented in Table \ref{tab:best_cases}. It demonstrates that McCormick MOT is significantly more effective when the option payoffs directly depend on the probabilities. A possible explanation is that the bicausality condition and its McCormick relaxations impose pointwise constraints, which may exclude certain events.
\begin{table}[H]
	\centering
	\begin{tabular}{ccccc}
		\hline
		Initial Date  & MOT Max &  MOT Min &  McCormick Max &  McCormick Min\\
		\hline
		02-28-2022 & 50.0031 &  0.0  & 2.3286 & 0.0  \\
		05-24-2022 & 22.0455 & 0.0 & 0.3480 & 0.0  \\
		08-25-2022 & 11.2728 & 0.0 & 0.0685 & 0.0 \\
		\hline
	\end{tabular}
	\caption{Prices obtained in the best cases for each initial date in Table \ref{tab:three}.}\label{tab:best_cases}
\end{table}

\begin{figure}
	\centering
	\begin{minipage}{0.45\textwidth}
		\centering
		\includegraphics[width=0.95\textwidth]{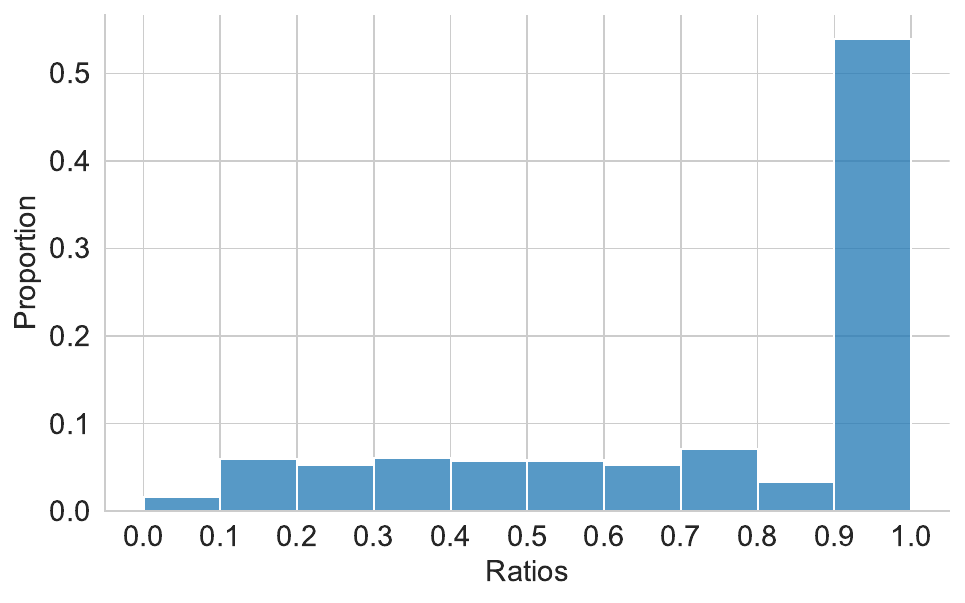}
		\subcaption{February 28, 2022}\label{Feb28}
	\end{minipage}
	\begin{minipage}{0.45\textwidth}
		\centering
		\includegraphics[width=0.95\textwidth]{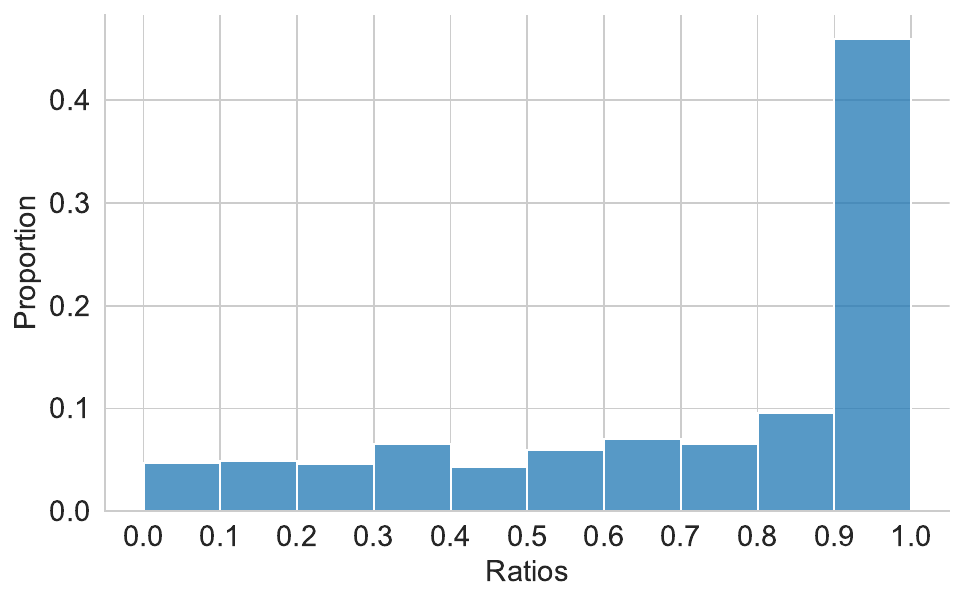}
		\subcaption{May 24, 2022}\label{May24}
	\end{minipage}
	\begin{minipage}{0.45\textwidth}
		\centering
		\includegraphics[width=0.95\textwidth]{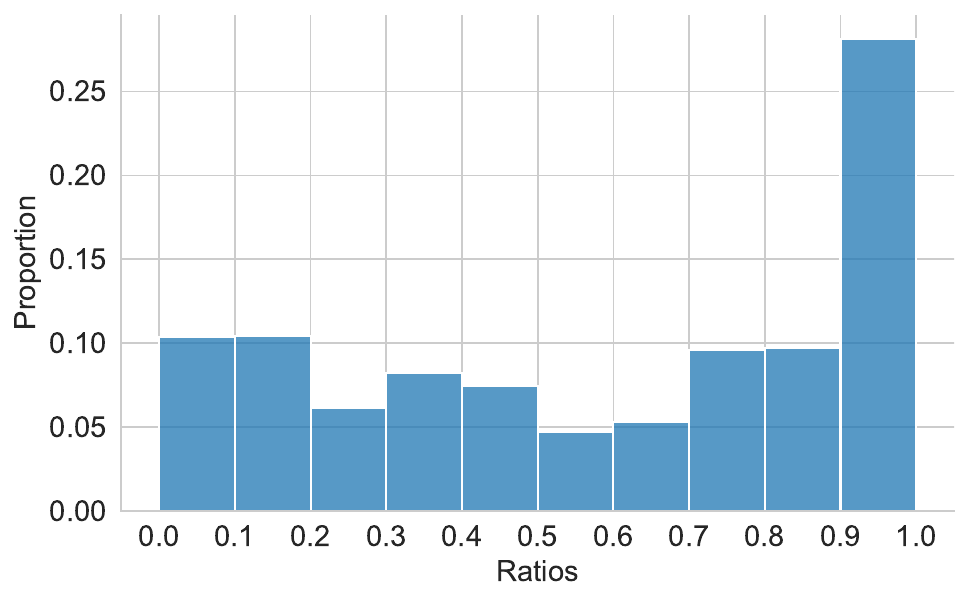}
		\subcaption{August 25, 2022}\label{Aug25}
	\end{minipage}%
	\caption{Distribution of price ratios on different dates, with JPM and MS as the underlying assets.}
\end{figure}

To further illustrate the behavior of McCormick MOT on digital payoffs, Figures \ref{Feb28}, \ref{May24}, and \ref{Aug25} show the distribution of price ratios for the three initial dates listed in Table \ref{tab:three}. While McCormick MOT may not always lead to price reductions, when improvements do occur, they are substantial. Moreover, in all three scenarios, the proportion of cases with price reductions exceeds $45.00\%$. 
%
%

To test on additional initial dates, we randomly select $100$ dates from our data. Due to computational constraints, only $50$ digital payoff cases are chosen randomly for each date. Table \ref{tab:random} presents the statistics of these $5,000$ ratios for each pair of underlying assets. Figures \ref{JPMrand} and \ref{GILDrand} show the corresponding ratio distributions. Notably, the distributions of the ratios are similar, regardless of whether the underlying stocks are more liquid or not.

\begin{table}[H]
	\centering
	\begin{tabular}{cccc}
		\hline
		Assets  & Ratios Mean &  Ratios Median & Ratios Min \\
		\hline
		JPM and MS 	& 0.7702 & 1.0 & 0.0045  \\
		GILD and GSK  & 0.7983 & 1.0 & 0.0045  \\
		\hline
	\end{tabular}
	\caption{Statistics of price ratios for digital options with randomly chosen dates.}\label{tab:random}
\end{table}

\begin{figure}[H]
	\centering
	\begin{minipage}{0.45\textwidth}
		\centering
		\includegraphics[width=0.95\textwidth]{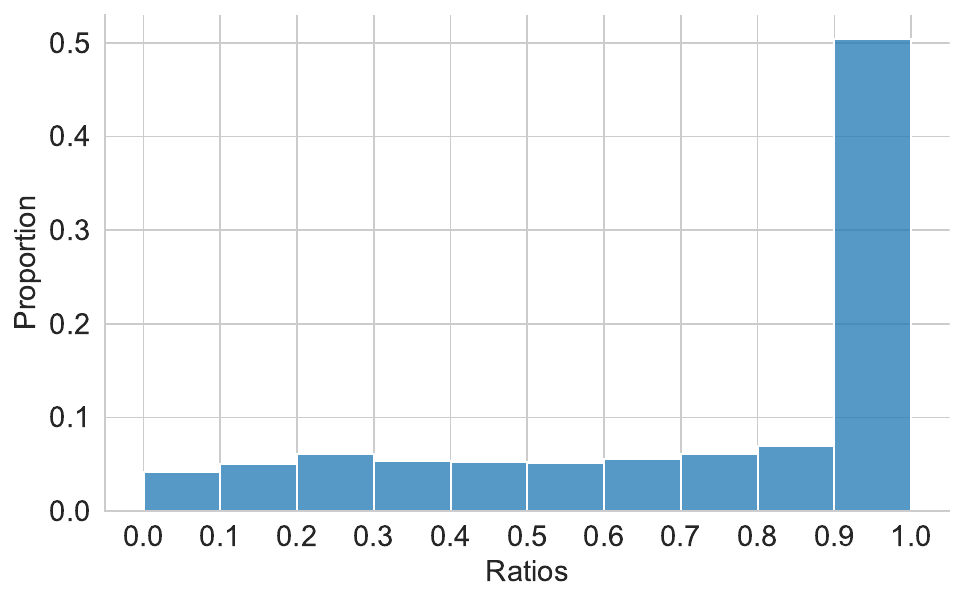}
		\subcaption{JPM and MS as the underlying assets}\label{JPMrand}
	\end{minipage}
	\begin{minipage}{0.45\textwidth}
		\centering
		\includegraphics[width=0.95\textwidth]{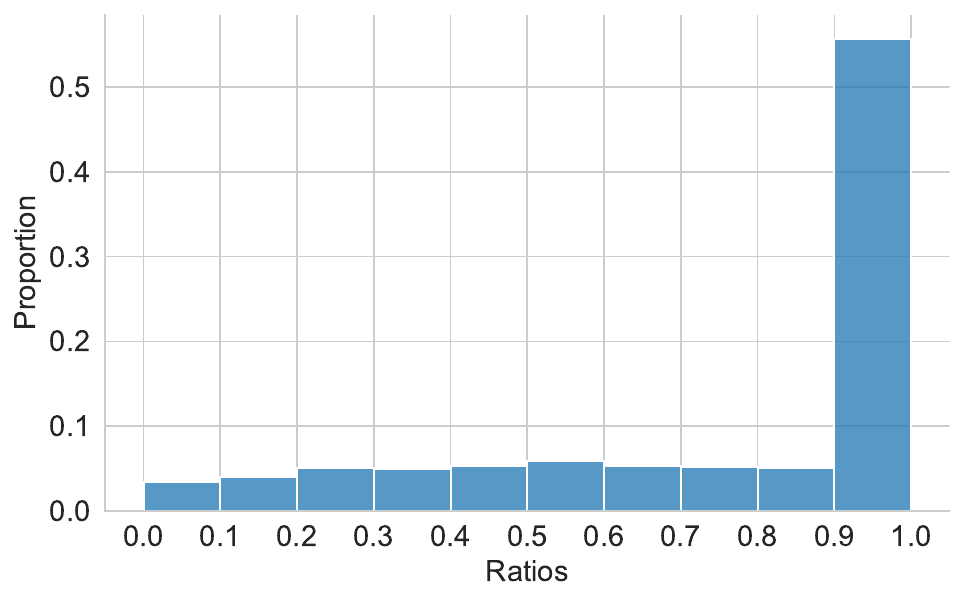}
		\subcaption{GILD and GSK as the underlying assets}\label{GILDrand}
	\end{minipage}%
	\caption{Distribution of price ratios.}
\end{figure}

%
%

\section*{Acknowledgments}
The authors would like to thank the anonymous referees and editors for their careful reading and valuable comments, which have greatly improved the manuscript. This research began when Bingyan Han was a postdoctoral researcher in the Department of Mathematics at the University of Michigan. He expresses gratitude to the University of Michigan for providing support and an atmosphere conducive to this work.


\appendix
\section{Proofs of results}

\begin{proof}[Proof of Proposition \ref{prop:equi}]
	Define $f^h(x_{1:N}): = \int h_t(y_{1:t})  \pi(dy_{1:t}|x_{1:N})$ with $h_t \in C_b(\cY_{1:t})$. It is worth noting that $f^h(x_{1:N})$ is defined $\pi$-almost surely. A coupling $\pi \in \Pi(\bar{\mu}, \bar{\nu})$ is causal if and only if
	\begin{equation}\label{eq:fh}
		f^h(x_{1:N}) = \int f^h(x_{1:t},\bar{x}_{t+1:N}) \pi(d\bar{x}_{t+1:N} | x_{1:t}), \quad \pi\text{-a.s.}
	\end{equation}
	for all $1 \leq t\leq N$ and all $f^h$ defined with $h_t \in C_b(\cY_{1:t})$. We emphasize that \eqref{eq:fh} holds $\pi$-almost surely, and the conditional kernel $\pi(d\bar{x}_{t+1:N} | x_{1:t})$ depends on the choice of $\pi$ since the joint distribution of $X_{1:N}$ is not fixed.
	
	Furthermore, \eqref{eq:fh} holds if and only if for any $g_t \in C_b (\cX_{1:N})$, we have
	\begin{align}
		\int g_t(x_{1:N})\left[ f^h(x_{1:N}) - \int f^h(x_{1:t},\bar{x}_{t+1:N}) \pi(d\bar{x}_{t+1:N} | x_{1:t}) \right] \pi(dx_{1:N}) = 0.
	\end{align}
	Interchanging the expectations on $x_{1:N}$ and $\bar{x}_{t+1:N}$ yields
	\begin{align}\label{eq:interchange}
		\int f^h(x_{1:N})  \left[ g_t(x_{1:N}) -\int g_t(x_{1:t}, \bar{x}_{t+1:N})  \pi(d\bar{x}_{t+1:N} | x_{1:t}) \right] \pi(d x_{1:N}) = 0.
	\end{align}
	Subsequently, utilizing the tower property, \eqref{eq:interchange} is equivalent to
	\begin{align}
		\int h_t(y_{1:t})\left [ g_t(x_{1:N}) - \int g_t(x_{1:t}, \bar{x}_{t+1:N}) \pi(d\bar{x}_{t+1:N} | x_{1:t}) \right ] \pi(dx_{1:N}, dy_{1:N}) = 0. 
	\end{align}
\end{proof}

\begin{proof}[Proof of Proposition \ref{prop:bicausal_primal}]
	In the finite discrete case with the Euclidean topology, the convergence of probability masses $\pi^n(x_{1:N}, y_{1:N})$ to $\pi(x_{1:N}, y_{1:N})$ is pointwise. Since $\cM_{bc}(\bar{\mu}, \bar{\nu})$ is characterized by linear and quadratic equalities, it is closed. Moreover, $\cM_{bc}(\bar{\mu}, \bar{\nu})$ is bounded and non-empty. By the Weierstrass theorem, the infimum is attained.
	
	For the second part, suppose that the problem has an optimizer $\bar{\pi}$ in the interior of $\cM_{bc}(\bar{\mu}, \bar{\nu})$; otherwise, there is nothing to prove. Moreover, we can find a point $\pi^*$ on the boundary $\partial \cM_{bc}$ that minimizes the Euclidean distance between $\partial \cM_{bc}$ and $\bar{\pi}$. Let $\pi^m := 2 \bar{\pi} - \pi^*$, such that $\bar{\pi}$ is the midpoint of the line segment joining $\pi^m$ and $\pi^*$.
	
	Crucially, $\pi^m$ is in $\cM_{bc}(\bar{\mu}, \bar{\nu})$. If not, there would exist another point between $\bar{\pi}$ and $\pi^m$ on $\partial \cM_{bc}$ that is closer to $\bar{\pi}$ than $\pi^*$, leading to a contradiction. 
	
	The objective values corresponding to these transport plans satisfy
	\begin{equation}\label{c1}
		\int c d\bar{\pi} = \frac{1}{2} \int c d\pi^m + \frac{1}{2} \int c d \pi^*.
	\end{equation}
	Since $\bar{\pi}$ is an optimizer, it follows that
	\begin{equation}\label{c2}
		\int c d\bar{\pi} \leq \int c d\pi^m \quad \text{and} \quad \int c d\bar{\pi} \leq \int c d \pi^*.
	\end{equation}
	Hence, from \eqref{c1}-\eqref{c2}, we obtain
	\begin{equation}
		\int c d\bar{\pi} = \int c d\pi^m \quad \text{and} \quad \int c d\bar{\pi} = \int c d \pi^*.
	\end{equation}
	Then $\pi^*$ on $\partial \cM_{bc}$ must be an optimizer of the problem.
\end{proof}

\begin{proof}[Proof of Theorem \ref{thm:primal}]
	If we substitute the cost $c$ with $c + C(1 + \sum^N_{i=1}|x_i| + \sum^N_{j=1}|y_j|)$, it only introduces a finite constant to the original problem, thanks to the finite first moments and marginal conditions. Consequently, we can assume $c \geq 0$ henceforth.  
	
	We endow $L^1(\R^N \times \R^N)$ with the weak topology, where $f_n$ converges to $f$ in the weak topology if and only if 
	\begin{equation*}
		\int f_n v d \lambda \rightarrow \int f v d \lambda, \quad \forall \, v \in L^\infty(\R^N \times \R^N).
	\end{equation*}
	Since every function in $\cM(\bar{g}, \bar{h}; l, u)$ is bounded by the same integrable function $u$, \citet[Theorem 4.7.20 (v)]{bogachev2007measure} shows that $\cM(\bar{g}, \bar{h}; l, u)$ has a compact closure under the weak topology of $L^1(\R^N \times \R^N)$. If we can demonstrate that $\cM(\bar{g}, \bar{h}; l, u)$ is closed, then it is also compact. 
	
	Consider a sequence of functions $\{f_n\}$ in $\cM(\bar{g}, \bar{h}; l, u)$ that converges to $f$ in the weak topology of $L^1(\R^N \times \R^N)$. This convergence is expressed as:
	\begin{equation*}
		\lim_{n \rightarrow \infty} \int f_n v d\lambda = \int f v d\lambda. 
	\end{equation*}
	First, select $v = \mathbf{1}_{A}$ with a measurable set $A$ in $\R^N \times \R^N$. Since it holds that 
	$$\int f_n(x_{1:N}, y_{1:N}) v(x_{1:N}, y_{1:N}) d\lambda \leq \int u(x_{1:N}, y_{1:N}) v(x_{1:N}, y_{1:N}) d\lambda,$$
	we can obtain that $f \leq u$ $\lambda$-a.e. Similarly, we can show $f \geq l$, $\lambda$-a.e. Next, consider $v(x_{1:N}, y_{1:N}) = \psi(x_1) \in L^\infty(\R)$. Then, $f(x_{1:N}, y_{1:N})$ has the marginal $g_1$ on $x_1$. Other marginal constraints can also be verified. To establish that $f$ satisfies the martingale constraint, we observe that
	\begin{align*}
		& \Big|\int_{\R^N \times \R^N} f_{n}(x_{1:N}, y_{1:N}) \alpha_t(x_{1:t}, y_{1:t}) (x_{t+1} - x_{t}) d\lambda \\
		& \qquad - \int_{\R^N \times \R^N} f(x_{1:N}, y_{1:N}) \alpha_t(x_{1:t}, y_{1:t}) (x_{t+1} - x_{t}) d\lambda \Big| \\
		& \leq \Big|\int_{K} f_{n}(x_{1:N}, y_{1:N}) \alpha_t(x_{1:t}, y_{1:t}) (x_{t+1} - x_{t}) d\lambda - \int_{K} f(x_{1:N}, y_{1:N}) \alpha_t(x_{1:t}, y_{1:t}) (x_{t+1} - x_{t}) d\lambda \Big|\\
		& \quad + \Big|\int_{(\R^N \times \R^N)\backslash K} f_{n}(x_{1:N}, y_{1:N}) \alpha_t(x_{1:t}, y_{1:t}) (x_{t+1} - x_{t}) d\lambda \\
		& \qquad \quad - \int_{(\R^N \times \R^N)\backslash K} f(x_{1:N}, y_{1:N}) \alpha_t(x_{1:t}, y_{1:t}) (x_{t+1} - x_{t}) d\lambda \Big|,
	\end{align*}
	where $K := [-a, a]^N \times [-a, a]^N$ is a compact set in $\R^N \times \R^N$, and $\alpha_t \in L^\infty$. With a given compact set $K$, the first term converges to zero since $x_t$ and $x_{t+1}$ are bounded on $K$. The second term is bounded by $\| \alpha_t \|_\infty \varepsilon$ since $f_n$ and $f$ have marginals $(\bar{g}, \bar{h})$ with finite first moments. Here, $\varepsilon$ is a generic constant such that 
	\begin{equation*}
		\int_{(\R^N \times \R^N)\backslash K} f(x_{1:N}, y_{1:N}) |x_{t}| d\lambda < \varepsilon.
	\end{equation*}
	This same constant $\varepsilon$ also bounds other similar integrals, where $f_n$ and/or $x_{t+1}$ are used instead. When the side width $a \rightarrow \infty$, we have $\varepsilon \rightarrow 0$. Consequently, $f$ also satisfies the martingale constraint. In summary, the limit $f \in \cM(\bar{g}, \bar{h}; l, u)$, and the compactness follows.
	
	The proof of the lower semicontinuity for the objective is standard. When $c$ is bounded, the functional
	\begin{equation*}
		f \mapsto \int c f d\lambda
	\end{equation*}
	is continuous, following from the definition of weak topology. In the case where $c$ is nonnegative but unbounded, there exists a sequence of bounded measurable functions $\{ c_n \}$ converging increasingly to $c$. Using the monotone convergence theorem, we have
	\begin{equation*}
		\int f c d\lambda = \int f \sup_n c_n d\lambda = \sup_n \int f c_n d \lambda,
	\end{equation*}
	demonstrating that $\int f c d\lambda$ is a supremum of continuous functionals and, consequently, l.s.c. The claim then follows as the infimum is attained for an l.s.c. functional on a compact set.
\end{proof}

\begin{proof}[Proof of Lemma \ref{lem:capa}]
	Similarly to the proof of \citet[Proposition 1.22]{villani2003topics}, we can substitute the cost $c$ for $c+ C(1+ |x| + |y|) + \varepsilon$ for an arbitrarily large $\varepsilon>0$. With the finite first moments and marginal conditions, this modification adds the same finite constant to both sides of the strong duality. The result for the modified cost implies the result for the original cost function $c$. Hence, we proceed with the assumption $c \geq \varepsilon>0$ in the subsequent analysis.  
	
	Define the functionals $\Xi$ and $\Theta$ as follows:
	\begin{equation}
		\begin{aligned}
			\Xi: & (b(x, y), \kappa(x, y)) \in C_b(\R \times \R) \times C_b(\R \times \R) \\
			& \mapsto \left\{\begin{array}{lll}
				& \int \phi(x) g(x) dx + \int \varphi(y) h(y) dy - \int \kappa(x, y) u(x, y) dxdy, \\
				& \text{ if } b(x, y) = \phi(x) + \varphi(y) \text{ and } \kappa(x, y) \leq 0; \\
				& + \infty, \text{ else. }\\ 
			\end{array}\right.
		\end{aligned}
	\end{equation}
	and
	\begin{equation}
		\begin{aligned}
			\Theta: & (b(x, y), \kappa(x, y)) \in C_b(\R \times \R) \times C_b(\R \times \R) \\
			& \mapsto \left\{\begin{array}{lll}
				&0, \text{ if } b(x, y) \geq -c(x, y) + \kappa(x, y) \text{ and } \kappa(x, y) \leq 0; \\
				& + \infty, \text{ else. }\\ 
			\end{array}\right.
		\end{aligned}
	\end{equation}
	$\Xi$ is well-defined. Specifically, if $\phi(x) + \varphi(y) = \tilde{\phi}(x) + \tilde{\varphi}(y)$ for all $x$ and $y$, then $\phi = \tilde{\phi} + s$ and $\varphi = \tilde{\varphi} - s$ with a constant $s$; see \citet[Section 1.1, p.27]{villani2003topics} for the same argument. Analogously, $\Theta$ is also well-defined.
	
	To verify the assumptions in Fenchel-Rockafellar duality theorem \cite[Theorem 1.9]{villani2003topics}, we observe that $\Xi$ and $\Theta$ are finite at $(b, \kappa) = (0, -1)$. Moreover, under conditions $\|b\|_\infty \leq \varepsilon$ and $\|\kappa + 1\|_\infty \leq \frac{1}{2}$, it follows that $\kappa \leq 0$ and $-c + \kappa \leq -\varepsilon \leq b$. Consequently, $\Theta(0, -1) = 0$, and therefore $\Theta$ is continuous at $(0, -1)$. 
	
	A direct calculation yields
	\begin{align*}
		& \inf_{(b, \kappa) \in C_b \times C_b} \Big[ \Theta(b, \kappa) + \Xi(b, \kappa) \Big] \\ 
		& = \inf_{(\phi, \varphi, \kappa) \in C_b \times C_b \times C_b} \Big[  \int \phi(x) g(x) dx + \int \varphi(y) h(y) dy - \int \kappa(x, y) u(x, y) dxdy \Big| \\
		& \hspace{4cm} \phi(x) + \varphi(y) \geq -c(x, y) + \kappa(x, y), \kappa \leq 0 \Big] \\
		& = - \sup_{(\phi, \varphi, \kappa) \in C_b \times C_b \times C_b} \Big[  \int \phi(x) g(x) dx + \int \varphi(y) h(y) dy - \int \kappa(x, y) u(x, y) dxdy \Big| \\
		& \hspace{4cm} \phi(x) + \varphi(y) \leq c(x, y) + \kappa(x, y), \kappa \geq 0 \Big],
	\end{align*}
	where in the last equality, we substitute $(\phi, \varphi, \kappa)$ with $(-\phi, -\varphi, -\kappa)$.
	
	A challenge arises due to the fact that the topological dual of $C_b(\R \times \R)$, denoted as $(C_b(\R \times \R))^*$, is greater than $M(\R \times \R)$, which represents the set of all finite signed measures. To calculate the Legendre-Fenchel transforms of $\Theta$ and $\Xi$, we take advantage of \citet[Lemmas 1.24 and 1.25]{villani2003topics} to deal with continuous linear functionals on $C_b(\R \times \R)$. In addition, we note that the topological dual of the product space $C_b(\R \times \R) \times C_b(\R \times \R)$ is the product of dual spaces, namely $(C_b(\R \times \R))^* \times (C_b(\R \times \R))^*$.
	
	With linear forms $(l_1, l_2) \in (C_b(\R \times \R))^* \times (C_b(\R \times \R))^*$, the Legendre-Fenchel transform of $\Theta$ is defined as follows:
	\begin{align*}
		& \Theta^*(-l_1, -l_2) \\
		& = \sup_{(b, \kappa) \in C_b \times C_b} \Big[ \ang{-l_1, b} + \ang{-l_2, \kappa}  \Big|  b(x, y) \geq -c(x, y) + \kappa(x, y), \; \kappa(x, y) \leq 0 \Big] \\
		& = \sup_{(b, \kappa) \in C_b \times C_b} \Big[ \ang{l_1, b - \kappa} + \ang{l_1 + l_2, \kappa}  \Big|  b(x, y) \leq c(x, y) + \kappa(x, y), \; \kappa(x, y) \geq 0 \Big],
	\end{align*}
	where $\ang{\cdot, \cdot}$ denotes the duality bracket \cite[p. xiv]{villani2003topics}.
	
	Next, we identify the conditions for $\Theta^*(-l_1, -l_2) < +\infty$ and compute $\Theta^*(-l_1, -l_2)$. If there exists $v \geq 0$ satisfying $\ang{l_1 + l_2, v} > 0$, then we can choose $b= \kappa = k v$ and send $k \rightarrow \infty$. Similarly, if there exists $v \geq 0$ and $\ang{l_1, v} < 0$, then we can consider $\kappa = 0$ and $b = -kv$ with $k \rightarrow \infty$. Therefore, $\Theta^*(-l_1, -l_2) < + \infty$ only when $l_1$ is nonnegative and $l_1 + l_2$ is nonpositive. 
	
	Furthermore, the Legendre-Fenchel transform of $\Xi$ is given by
	\begin{align*}
		\Xi^*(l_1, l_2) = \sup_{(b, \kappa) \in C_b \times C_b} \Big[ & \ang{l_1, b} + \ang{l_2, \kappa} - \int \phi(x) g(x) dx - \int \varphi(y) h(y) dy \\
		& + \int \kappa(x, y) u(x, y) dxdy \Big| b(x, y) = \phi(x) + \varphi(y) \text{ and } \kappa(x, y) \leq 0 \Big].
	\end{align*}
	To avoid $\Xi^*(l_1, l_2) = +\infty$, we first need
	\begin{equation}\label{eq:margin}
		\ang{l_1, \phi + \varphi} = \int \phi(x) g(x) dx + \int \varphi(y) h(y) dy.
	\end{equation}
	It is noteworthy that $l_1$ should be nonnegative, as argued above. Moreover, $l_1$ also acts continuously on the subset $C_0(\R \times \R)$ of $C_b(\R \times \R)$, where $C_0(\R \times \R)$ denotes the set of all continuous functions approaching 0 at infinity. Since the topological dual of $C_0(\R \times \R)$ is $M(\R \times \R)$, there exists a unique measure $\pi \in M(\R \times \R)$ such that
	\begin{equation*}
		\ang{l_1, v} = \int v(x, y) d\pi, \quad  \forall \, v \in C_0(\R \times \R). 
	\end{equation*}
	We can then decompose $l_1 = \pi + R$, where $R$ is a continuous linear functional with
	\begin{equation*}
		\ang{R, v} = 0, \quad  \forall \, v \in C_0(\R \times \R). 
	\end{equation*}
	Under the condition \eqref{eq:margin}, \citet[Lemma 1.25]{villani2003topics} shows that $R=0$, and $l_1 = \pi$ is a nonnegative measure with Borel marginals $gdx$ and $hdy$.
	
	Another requirement to avoid $\Xi^*(l_1, l_2) = +\infty$ is
	\begin{equation}\label{eq:l2}
		\ang{l_2, \kappa} + \int \kappa(x,y) u(x, y) dxdy \leq 0, \quad \forall \, \kappa(x, y) \leq 0, \, \kappa \in C_b(\R \times \R). 
	\end{equation}
	Since $l_1 + l_2$ is nonpositive and $l_1$ is a measure, then
	\begin{align}
		\ang{l_1 + l_2, v} = \int v(x, y) d \pi + \ang{l_2, v} \leq 0, \quad \forall \, v(x, y) \geq 0, \, v \in C_b(\R \times \R). 
	\end{align} 
	Together with \eqref{eq:l2}, setting $\kappa = - v$ yields
	\begin{equation}\label{eq:dom}
		\int v(x, y) d \pi \leq - \ang{l_2, v} \leq \int v(x, y) u(x, y) dxdy, \quad \forall \, v(x, y) \geq 0, \, v \in C_b(\R \times \R). 
	\end{equation}
	This implies that $\pi$ is absolutely continuous with respect to the Lebesgue measure $\lambda$ on the Borel $\sigma$-algebra of $\R \times \R$. Otherwise, suppose that the opposite is true. Then, there exists a Borel set $B$ such that $\lambda(B) = 0$ and $\pi(B) > 0$. Since $\pi$ is regular \cite[Proposition 7.17]{bertsekas1978stoch}, there exists a closed subset $F \subset B$ and $\pi(F) > \delta > 0$ for a small $\delta$. Similar to the proof of \citet[Proposition 7.18]{bertsekas1978stoch}, let $G_n =\{(x, y) \in \R \times \R| d((x, y), F) < 1/n \}$ with a sufficiently large $n$. By Urysohn's lemma, there exist continuous functions $0 \leq v_n \leq 1$ such that $v_n = 0$ on $G^c_n$ and $v_n = 1$ on $F$. Then
	\begin{align*}
		\pi(F) \leq \int v_n(x, y) d\pi \leq \int v_n(x, y) u(x, y) d\lambda \leq \|u\|_\infty \int v_n d\lambda \leq \|u\|_\infty \lambda(G_n),
	\end{align*}
	which implies
	\begin{align*}
		\pi(F) \leq \|u\|_\infty \lambda(\cap^{\infty}_{n=1} G_n) = \|u\|_\infty \lambda(F) \leq \|u\|_\infty \lambda(B) = 0,
	\end{align*}
	leading to a contradiction with $\pi(F) > \delta > 0$. By the Radon-Nikod\'ym theorem \cite[Theorem 9.3, p.98]{rogers1993diffusions}, the Radon-Nikod\'ym density exists: $f(x, y) = \frac{d\pi}{d\lambda},\, \lambda\text{-a.e.}$ 
	
	Furthermore, \eqref{eq:dom} implies $f(x, y) \leq u(x, y), \, \lambda\text{-a.e.}$ Otherwise, suppose $\lambda(B) > 0$, where $B := \{ (x, y) \in \R \times \R| f(x, y) > u(x, y)\}$. $B$ is a Borel set. Since $\lambda$ is regular, there exists a closed subset $F \subset B$ and $\lambda(F) > 0$. Introducing continuous and bounded functions $v_n$ similarly as before, \eqref{eq:dom} and the dominated convergence theorem lead to
	\begin{align*}
		\int \textbf{1}_F f d\lambda \leq \int v_n fd\lambda \leq \int v_n u d\lambda \leq \int \textbf{1}_{G_n} u d\lambda \rightarrow \int \textbf{1}_{F} u d\lambda,
	\end{align*}
	which contradicts the fact that $f > u$ on $F$ and $\lambda(F) > 0$.
	
	In summary, $\Theta^*(- l_1, -l_2) + \Xi^*(l_1, l_2) < +\infty$ requires that $l_1$ is a probability measure in $\Pi(g, h)$ with density $0 \leq f \leq u$ and $l_1 + l_2$ is nonpositive. Clearly, the reverse direction is also true: If $l_1$ is a probability measure in $\Pi(g, h)$ with density $0 \leq f \leq u$ and $l_1 + l_2$ is nonpositive, then $\Theta^*(- l_1, -l_2) + \Xi^*(l_1, l_2) < +\infty$. 
	
	Moreover, when $l_1$ is a probability measure in $\Pi(g, h)$ with density $0 \leq f \leq u$ and $l_1 + l_2$ is nonpositive, it follows that $\Xi^*(l_1, l_2) = 0$ and
	\begin{align*}
		\Theta^*(-l_1, -l_2) &= \sup_{b \in C_b} \Big[ \int b(x, y) f(x,y) dxdy \Big|  b(x, y) \leq c(x, y) \Big] = \int c(x, y) f(x,y) dxdy,
	\end{align*}
	because the l.s.c. cost $c \geq \varepsilon$ can be approximated pointwise by a monotonically increasing sequence of continuous and bounded functions. Then the equality above follows from the monotone convergence theorem. Therefore, the right-hand side of the Fenchel-Rockafellar duality reduces to
	\begin{align*}
		& \max_{(l_1, l_2) \in (C_b)^* \times (C_b)^*} \Big[ - \Theta^*(- l_1, -l_2) - \Xi^*(l_1, l_2) \Big] 
		& =  - \min_{f \in \Pi(g, h; u)} \int c(x, y) f(x,y) dxdy.
	\end{align*}
	We have proved the strong duality \eqref{eq:capa_dual}, albeit with continuous and bounded $(\phi, \varphi, \kappa)$. In the $L^1$ case, the claim follows from
	\begin{align*}
		\inf_{f \in \Pi(g, h; u)} P(f) = \sup_{\Psi_c \cap C_b} F(\phi, \varphi, \kappa) \leq \sup_{\Psi_c} F(\phi, \varphi, \kappa) \leq \inf_{f \in \Pi(g, h; u)} P(f),
	\end{align*}
	where the last inequality is a result of weak duality, similar to \citet[Proposition 1.5]{villani2003topics}. In addition, in the first equality, the notation $\Psi_c \cap C_b$ indicates that $(\phi, \varphi, \kappa)$ are also continuous and bounded.
\end{proof}

\begin{proof}[Proof of Theorem \ref{thm:dual}]
	Similarly to the proof of Lemma \ref{lem:capa}, we can assume the cost $c \geq 0$. 
	
	Consider a continuous and bounded cost $c$. Maximization in the dual problem can be achieved in two steps. First, we maximize over $(\phi, \varphi, \kappa)$ when other multipliers $(\alpha, \beta, \gamma, \eta, \theta)$ are given. This reduces to a dual problem with the upper capacity constraint in Lemma \ref{lem:capa}, but with a new cost given by
	\begin{equation}
		\begin{aligned}
			& \tilde{c}(x_{1:N}, y_{1:N}, \alpha, \beta, \gamma, \eta, \theta;u, l) \\
			&=  c(x_{1:N}, y_{1:N}) - \theta(x_{1:N}, y_{1:N})\\
			&\quad + \sum^{N-1}_{t=1} \alpha_t(x_{1:t}, y_{1:t}) (x_{t+1} - x_t) + \sum^{N-1}_{t=1} \beta_t(x_{1:t}, y_{1:t}) (y_{t+1} - y_t) \\
			&\quad + \Psi(x_{1:N}, y_{1:N}, \gamma, \eta; u, l).
		\end{aligned}
	\end{equation}
	
	On the basis of our assumptions, we have $c, \theta, \alpha, \beta, \gamma, \eta, u, l \in C_b$. Hence, $\tilde{c}$ is continuous, and $|\tilde{c}| \leq C(1 + \sum^N_{t=1} |x_t| + \sum^N_{t=1} |y_t|)$. An application of Lemma \ref{lem:capa} leads to
	\begin{align}
		D & =  \sup_{\substack{\alpha_t, \beta_t \\ \gamma_{t, i}, \, \eta_{t, i} \geq 0 \\ \theta \geq 0}} \sup_{\substack{\Phi_c \text{ with } (\alpha_t, \beta_t, \gamma_{t, i}, \, \eta_{t, i}, \theta) \text{ given }}} D(\phi, \varphi, \alpha, \beta, \gamma, \eta, \kappa, \theta) \\
		& = \sup_{\substack{\alpha_t, \beta_t \\ \gamma_{t, i}, \, \eta_{t, i} \geq 0 \\ \theta \geq 0}} \inf_{f \in \Pi(\bar{g}, \bar{h}; u)} \int \tilde{c}(x_{1:N}, y_{1:N}, \alpha, \beta, \gamma, \eta, \theta;u, l) f(x_{1:N}, y_{1:N}) d \lambda \\
		& = \inf_{f \in \Pi(\bar{g}, \bar{h}; u)} \sup_{\substack{\alpha_t, \beta_t \\ \gamma_{t, i}, \, \eta_{t, i} \geq 0 \\ \theta \geq 0}} \int \tilde{c}(x_{1:N}, y_{1:N}, \alpha, \beta, \gamma, \eta, \theta;u, l) f(x_{1:N}, y_{1:N}) d \lambda \label{eq:minimax} \\
		& = \inf_{f \in \cM(\bar{g}, \bar{h}; l, u)} \int c(x_{1:N}, y_{1:N}) f(x_{1:N}, y_{1:N}) d \lambda = P. \label{eq:inf_P}
	\end{align}
	\eqref{eq:minimax} follows from the minimax theorem, given in \citet[Theorem 45.8, p. 239]{strasser1985} or \citet[Theorem 3.1]{beiglbock2013model}. In fact, $\Pi(\bar{g}, \bar{h}; u)$ is compact in the weak topology of $L^1$. Since $\tilde{c}$ has a linear growth rate and the marginals have finite first moments, we can also prove that the objective $\int \tilde{c} f d\lambda$ is continuous in $f \in \Pi(\bar{g}, \bar{h}; u)$ using the weak topology of $L^1$, similar to the approach in \citet[Lemma 2.2]{beiglbock2013model}.
	
	The final equality \eqref{eq:inf_P} is derived from the observation that any violation of the martingale condition, the McCormick relaxations, or the lower bound results in an infinite value.
	
	In the case of a nonnegative and l.s.c. cost function $c$, the strong duality can be established through a standard approximation argument. For example, see the last part of the proof presented in \citet[Theorem 1.1]{beiglbock2013model}. 
\end{proof}

\begin{proof}[Proof of Proposition \ref{prop:cali}]
	The proof is straightforward. Clearly, the objective and constraints in \eqref{eq:mu_obj}-\eqref{eq:prob} are linear. If the optimal value is finite, then the optimizer $\mu$ exists. The constraints \eqref{eq:price}-\eqref{eq:prob} ensure that $\mu$ is a risk-neutral measure, yielding forward prices $F_t$ and option prices $c_{i, t}$. By the property of absolute values, the smallest value of $|c_{i, t} - a_{i, t}| + | c_{i, t} - b_{i, t}|$ is given by $|a_{i, t} - b_{i, t}|$, which implies that $c_{i, t} \in [b_{i, t}, a_{i, t}]$. The claim follows.
\end{proof}

\section{Additional details on the derivation of the dual problem \eqref{prob:dual}}\label{appen:derive_dual}

With the Lagrange multipliers introduced in Section \ref{sec:dual}, the Lagrangian should include the following components:
\begin{enumerate}[label=(\arabic*)]
	\item Terms related to the marginal constraints, given by
	\begin{equation*}
		\int \phi_t(x_t) (g_t(x_t) - f(x_t)) \lambda(dx_t) \quad \text{ and } \quad \int \varphi_t(y_t) (h_t(y_t) - f(y_t)) \lambda(dy_t);
	\end{equation*}
	
	\item Terms related to the martingale constraints, given by
	\begin{equation*}
		\int f(x_{1:N}, y_{1:N}) \alpha_t(x_{1:t}, y_{1:t}) (x_{t+1} - x_t) d\lambda
	\end{equation*}
	and 
	\begin{equation*}
		\int f(x_{1:N}, y_{1:N}) \beta_t(x_{1:t}, y_{1:t}) (y_{t+1} - y_t) d\lambda;
	\end{equation*}
	
	\item Terms related to the McCormick relaxations. For simplicity, we provide two examples, with others obtained similarly. The term related to the first inequality in the McCormick relaxation of the causality condition is
	\begin{align*}
		\int \gamma_{t, 1}(x_{1:N}, y_t)\big( & -w(x_{1:N}, y_t) + l(x_{1:N}, y_t) f(x_{1:t}) \\
		&  + f(x_{1:N}, y_t) l(x_{1:t}) - l(x_{1:N}, y_t) l(x_{1:t}) \big) \lambda(dx_{1:N}, dy_t).
	\end{align*}
	The term related to the first inequality in the McCormick relaxation of the anticausality condition is
	\begin{align*}
		\int \eta_{t, 1}(x_t, y_{1:N})\big( & -\tilde{w}(x_t, y_{1:N}) + l(x_t, y_{1:N}) f(y_{1:t}) \\
		&  + f(x_t, y_{1:N}) l(y_{1:t}) - l(x_t, y_{1:N}) l(y_{1:t}) \big) \lambda(dx_t, dy_{1:N}),
	\end{align*}
	where $\tilde{w}(x_t, y_{1:N})$ is the anticausal counterpart of $w(x_{1:N}, y_t)$;
	
	\item Terms related to the upper and lower bounds, given by
	\begin{equation*}
		\int \kappa(x_{1:N}, y_{1:N}) ( f(x_{1:N}, y_{1:N}) - u(x_{1:N}, y_{1:N})) d\lambda
	\end{equation*}
	and
	\begin{equation*}
		\int \theta(x_{1:N}, y_{1:N}) (l(x_{1:N}, y_{1:N}) - f(x_{1:N}, y_{1:N})) d\lambda;
	\end{equation*}
	
	\item The primal objective itself, i.e., $\int c(x_{1:N}, y_{1:N}) f(x_{1:N}, y_{1:N}) d\lambda$.
\end{enumerate}
Heuristically, with the Lagrangian above, we interchange the infimum with the supremum, group terms by $w(x_{1:N}, y_t)$, $\tilde{w}(x_t, y_{1:N})$, and $f(x_{1:N}, y_{1:N})$, and finally obtain the dual problem in \eqref{prob:dual}.

\end{document}